\documentclass[nonblindrev]{arxiv} 

\usepackage{natbib}
 \bibpunct[, ]{(}{)}{,}{a}{}{,}%
 \def\bibfont{\small}%
\TheoremsNumberedThrough     

\EquationsNumberedThrough    

\usepackage{bbm}
\usepackage{xspace}
\usepackage{comment}
\usepackage[ruled,vlined,linesnumbered]{algorithm2e}
\usepackage{bbm}
\usepackage[mathcal]{euscript}
\usepackage{soul}
\usepackage{mathrsfs}
\usepackage{marginnote}
\usepackage{mathtools}
\usepackage{subcaption}
\usepackage{amsmath,amssymb}
\usepackage[normalem]{ulem}
\usepackage{soul}
\usepackage{hyperref}
\newtheorem{mydefinition}{Definition}
\usepackage{xcolor}

\DeclareMathOperator{\Ex}{\mathbb{E}}

\newcommand{\slots}{\mathscr{T}\xspace}
\newcommand{\slot}{j}

\newcommand{\components}{\mathscr{C}\xspace}
\newcommand{\component}{c}

\newcommand{\nonl}{\renewcommand{\nl}{\let\nl\oldnl}}

\newcommand{\numcomponents}{n}

\newcommand{\utility}{r}
\newcommand{\fullutility}{\mu}
\newcommand{\utilities}{R}
\newcommand{\knapsackcapacity}{W}
\newcommand{\cardinalitycapacity}{C}

\newcommand{\Problem}{\texttt{CCKP}}

\newcommand{\Approx}{\texttt{APPROX}}

\newcommand{\numitems}{n}

\newcommand{\maxcapacity}{\omega}

\newcommand{\vecx}{\boldsymbol{x}}
\newcommand{\vecy}{\boldsymbol{y}}
\newcommand{\vecz}{\boldsymbol{z}}

\newcommand{\online}{\mathcal{A}}

\begin{document}


\RUNAUTHOR{}

\RUNTITLE{Cardinality-Constrained Continuous Knapsack Problem with Concave Piecewise-Linear Utilities}

\TITLE{Cardinality-Constrained Continuous Knapsack Problem with Concave Piecewise-Linear Utilities}

\ARTICLEAUTHORS{%
\AUTHOR{Miao Bai}
\AFF{Department of Operations and Information Management, University of Connecticut,  \EMAIL{miao.bai@uconn.edu}, \URL{}}
\AUTHOR{Carlos Cardonha}
\AFF{Department of Operations and Information Management, University of Connecticut, \EMAIL{carlos.cardonha@uconn.edu}, \URL{}}
} 

\ABSTRACT{
We study an extension of the cardinality-constrained knapsack problem wherein each item has a concave piecewise linear utility structure  (\Problem{}), which is motivated by applications such as resource management problems in monitoring and surveillance tasks. Our main contributions are combinatorial algorithms for the offline \Problem{} and an online version of the \Problem{}. For the offline problem,  we present a fully polynomial-time approximation scheme and show that it can be cast as the maximization of a submodular function with cardinality constraints; the latter property allows us to derive a greedy $(1 - \frac{1}{e})$-approximation algorithm. For the online \Problem{} in the random order model, we derive a
$\frac{10.427}{\alpha}$-competitive algorithm based on~$\alpha$-approximation algorithms for the offline~\Problem{}; moreover, we derive stronger guarantees for the cases wherein the cardinality capacity is very small or relatively large. Finally, we investigate the empirical performance of the proposed algorithms in numerical experiments. 
}


\maketitle

\section{Introduction}

In the \textit{continuous knapsack problem}, the goal is to maximize the total utility of fully or fractionally selecting items from a given set subject to a knapsack constraint, i.e., the total weight of the utilized items must not exceed a given weight capacity \citep{dantzig1957discrete}. Although its classic form can be efficiently solved \citep{dantzig1957discrete}, \citet{de2003polyhedral} characterize an important NP-hard variant -- the \textit{cardinality-constrained} continuous knapsack problem, which includes additional constraints on the number of picked items, as seen in many real life applications \citep{Kim2019}.

It is worth noting that existing studies on the cardinality-constrained continuous knapsack problem have focused on linear utility functions for picking items \citep{de2003polyhedral,Pienkosz2017}. However, this assumption on the form of utility functions may not be applicable in many real-life scenarios, including resource management problems in monitoring and surveillance tasks. For example, sensors (e.g., surveillance cameras and network sniffers) are widely used to monitor physical activities across multiple locations within a region of interest or to monitor devices on a wireless network \citep{fang2012strategizing}. However, restrictions stemming from sensor availability, budget constraints, and limited storage and processing capabilities impose constraints on both the number of sensors to use (i.e., the cardinality constraint) and the total load of data collected by all sensors (i.e., the knapsack constraint) in these tasks \citep{Xian2018, Ren2022,granmo2007learning}. As a result, decision-makers need to determine the locations or devices to monitor (i.e., items to pick) and control the data load from each sensor (i.e., the amount to pick from each item) by adjusting their sampling or scan frequencies \citep{granmo2007learning, Gomez2022}. Moreover, the utility or value of information collected by each sensor is better seen as a concave function of their data load \citep{granmo2007learning}, reflecting the intuitive notion that the marginal value of data decreases as a better knowledge of the monitored target's status is built based on previous data.

In addition to the \textit{offline setting}, where monitoring decisions are made based on complete information about all items (e.g., targets of interest), decision scenarios may also arise in an \textit{online setting}. In this context, items sequentially become known to decision-makers, who must make irrevocable decisions about picking an item (e.g., whether and how much data to collect from a target) before the next one arrives \citep{marchetti1995stochastic}. Despite its important applications, the online continuous knapsack problem has received limited attention in the literature~\citep{giliberti2021improved}.

In this work, we study a cardinality-constrained continuous knapsack problem with \textit{concave piecewise linear utility functions} (\Problem{}) in  \textit{offline} and \textit{online} settings. The concave piecewise linear utility function can approximate general concave nonlinear utility functions, which connects \Problem{} to many online and offline decision-making processes in economics and business applications. Our work focuses on the design of approximation algorithms for the \Problem{}, as the problem is NP-hard. For the offline~\Problem{}, we present a fully polynomial-time approximation scheme (FPTAS) and a greedy~$(1 - \frac{1}{e})$-approximation algorithm; we show in our numerical experiments that the latter has a much better empirical performance than its theoretical guarantee. Finally, we present an algorithm for the online~\Problem{} in the random order model. The algorithm is iterative and relies on the solution of instances of the offline~\Problem{}. We show that we can use any $\alpha$-approximation algorithm for the offline~\Problem{} to solve these sub-problems at the expense of extending the loss factor~$\alpha$ to the competitive ratio of the online algorithm. More precisely, by using an efficient~$\alpha$-approximation algorithm for the offline~$\Problem{}$, we obtain a $\frac{10.427}{\alpha}$-competitive algorithm for the online~\Problem{}. We refine our analysis to obtain stronger guarantees when the cardinality capacity is very small or relatively large. Finally, we discuss cases where the utility functions are concave but not necessarily piecewise linear.

\subsection{Contributions to the literature}\label{sec:relatedwork}
Among numerous studies on knapsack problems, our work primarily contributes to the literature on approximation algorithms for knapsack problems; we refer interested readers to \citep{Cacchiani2022a,Cacchiani2022} for a comprehensive review of the area, model formulations, and exact solutions. We investigate a version of the continuous knapsack problem with item-wise independent piecewise linear concave utility functions. Previous literature studies efficient exact algorithms for concave mixed-integer problems with connection to the knapsack problem~\citep{more1990solution,bretthauer1994algorithm,sun2005exact,yu2017polyhedral}. For knapsack problems with convex utilities, we refer to~\citet{levi2014continuous} and~\citet{halman2014fully}. Finally, for a survey on techniques for other nonlinear knapsack problems, we refer to~\citet{bretthauer2002nonlinear}.

For the offline version of the~\Problem{}, we present an FPTAS. Regarding computational complexity, the~\Problem{} bears greater resemblance to the 0-1 knapsack problem than it does to the continuous knapsack problem. Among the extensive literature on approximation schemes for the knapsack problem \citep{sahni1975approximate,ibarra1975fast,lawler1979fast}, the work by \citet{hochbaum1995nonlinear} is the most relevant as it presents approximation schemes for both integer and continuous versions of the knapsack problem with concave piecewise linear utility function and convex piecewise linear weight function. Different from their work, we consider a cardinality constraint, which increases the complexity of the problem \citep{de2003polyhedral}. For the 0-1 knapsack problem with cardinality constraint, \citet{caprara2000approximation} present the first approximation schemes (both a PTAS and an FPTAS), later improved by~\citet{mastrolilli2006hybrid} and~\citet{li2022faster}. Departing from these works that assume linear utility functions and focus on offline settings, we consider concave utility functions and study the online version of the problem.

The online~\Problem{} belongs to the broad class of online packing problems, which have been intensively investigated in the literature. Some of these problems admit online algorithms with constant-factor competitive ratios \citep{karp1990optimal,mehta2007adwords,buchbinder2009online}; in particular, \citet{devanur2012online} derive strong guarantees for fractional packing problems where the items have concave utilities. However, the theoretical guarantees for knapsack problems in fully adversarial online settings are weaker. In particular, it follows from early results in the literature that the online~\Problem{} does not admit a~$c$-competitive algorithm for any constant~$c$ if an adversary has the power to define the weights and utility functions of the items and their arrival sequence, similar to the online 0-1 knapsack problem~\citep{marchetti1995stochastic,zhou2008budget}. Some early work in the online packing literature has focused on algorithms with competitive factors parameterized by features of the problem, such as sparsity  \citep{azar2016online,kesselheim2018primal}, eventually allowing for a controlled level of infeasibility \citep{buchbinder2009online}. Moreover, different online models have been studied, such as bandits (when resource consumption and utility are unknown) and the i.i.d. request model (when features of the items are sampled with repetition from an unknown distribution) \citep{badanidiyuru2018bandits,devanur2019near}. 

In this work, we study the online~\Problem{} in the \textit{random order model}, in which an adversary sets the utility functions and weights, but the number of items is known a priori, and their arrival order is drawn uniformly at random across all permutations of the items \citep{babaioff2007knapsack,kesselheim2018primal, albers2021improved}. Unlike the classic definition of competitive ratio for online problems \citep{borodin2005online}, the performance of online algorithms in the random order model is inherently stochastic. Therefore, the competitive ratio for these cases considers expected outcomes and is defined as follows.
\begin{mydefinition}[\textbf{Competitive Ratio}] \label{def:ratio}
An online algorithm $\mathcal{A}$
is $c$-competitive for a maximization problem defined over a set of instances~$\mathcal{I}$ and for some $c\geq 1$ if \[
\Ex[\mathcal{A}(I)]\geq \left(\frac{1}{c}-o(1)\right) O^*(I)
\]
holds for all instances $I\in \mathcal{I}$, where $\Ex[\mathcal{A}(I)]$ is the expected objective value of~$\mathcal{A}$ for instance~$I$,~$O^*(I)$ is the optimal offline objective value for~$I$, and $o(1)$ is asymptotic in the number of items.
\medskip
\end{mydefinition}

A few online problems in the random order model are closely related to the~\Problem{}. In the secretary problem \citep{ferguson1989solved}, items arrive sequentially in random order, and their utility is only revealed upon arrival; one must make a one-time irrevocable decision to pick an item to maximize the expected utility. \citet{babaioff2007knapsack} are among the first to study the 0-1 knapsack problem in the random order model; they present a $10e$-competitive algorithm for the problem. \citet{kesselheim2018primal} later obtain an 8.06-competitive algorithm, and the state-of-the-art is the 6.65-competitive algorithm by~\citet{albers2021improved}. To our knowledge, \citet{karrenbauer2020reading} is the first article to investigate the continuous knapsack problem in the random order model; they present a 9.37-competitive algorithm for the problem. The current state of the art is the 4.39-competitive algorithm proposed by~\citet{giliberti2021improved}.

We contribute to the literature by studying an important variant of the online continuous knapsack problem. The~\Problem{} incorporates an additional packing dimension (i.e., the cardinality capacity), so these early results do not extend to our problem. Nevertheless, our algorithm and analysis explore some ideas presented in these papers. In particular, similar to~\citet{albers2021improved} and~\citet{giliberti2021improved}, we use a sequential algorithm, which changes its behavior over time. With respect to the analysis, we re-use some results associated with the knapsack problem from~\citet{albers2021improved} and explore the uniformity of the arrival orders to extract nontrivial lower bounds based on the strategy used  in~\citet{kesselheim2018primal} and~\citet{albers2021improved} to analyze algorithms for the generalized assignment problem. Finally, forced by the structure (or lack thereof) of the~$\Problem{}$, we introduce new elements to our algorithm and analysis to improve the competitive ratio, such as a scaling factor to control the knapsack space allocated to each item.

Different from the 0-1 knapsack problem, the utilization of an item in the~$\Problem{}$ is not binary but a choice of the algorithm. Moreover, the utility density of an item is not always constant in our case, as the utility functions are concave. Therefore, different from~\citet{albers2021improved}, we must sacrifice utility and scale down the amount of knapsack capacity allotted to each item to gain control over knapsack utilization. Moreover, \citet{giliberti2021improved} explore the fact that, in the continuous knapsack problem, the optimal utilization of an item when the problem is solved for all items of a set~$\slots$ is a lower bound for the optimal utilization of the same item when the problem is restricted to a subset of~$\slots$ that contains~$\slot$. Therefore, the algorithm by~\citet{giliberti2021improved} will only under-allocate an item if there is not enough residual knapsack capacity. This observation plays a key role in the analysis presented in~\cite{giliberti2021improved}, but it does not hold for the~$\Problem{}$. Therefore, our analysis is heavily based on the uniformity of the arrival rates instead, as in~\cite{kesselheim2018primal} and \citet{albers2021improved}. 

Finally, \citet{kesselheim2018primal} and \citet{albers2021improved} rely on the solutions of computationally tractable relaxations of their underlying problems when deciding to incorporate an item. In our case, we must use solutions of~$\alpha$-approximation algorithms for the problem; as a consequence, the approximation error~$\alpha$ propagates to the competitive factor of our online algorithm. Finally, our analysis and results extend to cases where the utility functions are continuous (i.e., not necessarily piecewise linear) and satisfy smoothness conditions (e.g., Lipschitz continuous).

\subsection{Organization}

The manuscript is organized as follows. Section~\ref{sec:definition} introduces the notation. Section~\ref{sec:Approximation Algorithms} presents the approximation results for the offline version of the~\Problem{}. Section~\ref{sec:online} investigates the online~\Problem{} in the random order model. Section~\ref{sec:experiments}  reports the results of our numerical experiments. Finally, we conclude the paper in Section~\ref{sec:conclusions}. Results and proofs omitted from the main text are presented in the Appendix.

\section{Definition and Notation}\label{sec:definition}

Let~$\slots$ denote a set of~$\numitems$ items; each item~$\slot$ has a capacity $\maxcapacity_{\slot}$ and is associated with a concave piecewise linear utility function~$\utilities_{\slot}(\cdot)$. We represent solutions to the~\Problem{} using a vector $\vecx:\slots \rightarrow \mathbb{R}^+$ to denote the utilization~$x_{\slot}$ of each item~$\slot$. A feasible solution must satisfy a knapsack and a cardinality constraint, parameterized by~$\knapsackcapacity$ and~$\cardinalitycapacity$, respectively. The objective is to maximize the overall utility of utilized items. The~$\Problem$ admits the following mathematical programming formulation \textbf{knapsack constraint needs a coefficient}:
\[
\begin{array}{ccll}\label{model:original}
\tag{\textbf{\texttt{CKP-o}}}
&\max
& \sum\limits_{\slot \in \slots}\utilities_{\slot} (x_{\slot})  \\
&(a) & \sum\limits_{\slot \in \slots}  x_{\slot}   \leq \knapsackcapacity \\
&(b)& \sum\limits_{\slot \in \slots} \| x_{\slot} \|_0 \leq \cardinalitycapacity  & \\
&(c)& x_{\slot} \leq \maxcapacity_{\slot} &\forall \slot \in \slots \\
&& x_{\slot} \in \mathbb{R}^+ &\forall \slot \in \slots
\end{array}
\]
The knapsack constraint \ref{model:original}-(a) limits the total utilization of items to~$\knapsackcapacity$. The cardinality constraint \ref{model:original}-(b) enforces an upper limit~$\cardinalitycapacity$ on the number of utilized items ($\| x_{\slot} \|_0=0$ if $x_{\slot}=0$ and $\| x_{\slot} \|_0=1$ if $x_{\slot}>0$). The capacity constraint \ref{model:original}-(c) limits the maximum utilization level of each item~$\slot$ by its capacity~$\maxcapacity_{\slot}$.

We make the following assumptions, which hold without loss of generality. First, the utilization level of an item is never larger than the knapsack capacity~$\knapsackcapacity$, so we assume that $\maxcapacity_{\slot} \leq \knapsackcapacity$ for all items~$\slot$. Moreover, if~$\maxcapacity_{\slot}'$ is the maximum value of $\overline{\maxcapacity}\in[0,\maxcapacity_{\slot}]$  such that~$R_{\slot}(\cdot)$ is non-decreasing in the interval~$[0,\overline{\maxcapacity}]$, we must have optimal solutions to~\ref{model:original} satisfying ~$x_{\slot}\leq \maxcapacity_{\slot}'$; otherwise, we can improve a solution with~$x_{\slot} > \maxcapacity_{\slot}'$ by reducing $x_{\slot}$ to~$\maxcapacity_{\slot}'$. Therefore, we can assume that all utility functions are non-decreasing. Finally, we assume that no two items have the same total utilities~$R_{\slot}(\maxcapacity_{\slot})$; if this condition does not apply, it suffices to apply a consistent tie-breaking procedure (e.g., increase the utilities by a negligibly small random amount) between elements whose total utilities coincide (\cite{babaioff2007knapsack}).

\section{Offline~\Problem{} and Approximation Algorithms}\label{sec:Approximation Algorithms}

The~$\Problem$ is a generalization of the cardinality-constrained knapsack problem, which was shown to be NP-hard by~\citet{de2003polyhedral}. Therefore, it follows that the~$\Problem$ is NP-hard.  In this section, we develop two approximation algorithms for the~$\Problem{}$ based on a structural property of the problem, presented in~\S\ref{sec:piecewise}. We present an FPTAS in~\S\ref{sec:FPTAS1}, and a greedy algorithm in~\S\ref{sec:submodularity}, which we prove to be a $\left(1 - \frac{1}{e}\right)$-approximation algorithm to the~$\Problem{}$.  Lastly, we conclude this section with a discussion of the case in which the utility functions are not piecewise linear.

\subsection{Utility Decomposition and Discreteness of Optimal Solutions}\label{sec:piecewise}

The piecewise linear structure allows us to decompose the utility function of each item~$\slot$ into a sequence of~$\numcomponents_{\slot}$ \textit{components}, each corresponding to a segment of the piecewise linear representation of~$\utilities_{\slot}(\cdot)$. The~$i^{th}$ component of item $\slot$ is denoted by~$\component_{\slot,i}$ and has capacity~$\maxcapacity_{\slot,i}$ and utility~$\utility_{\slot,i}$ per unit of utilization; because of the concavity of $\utilities_{\slot}(\cdot)$, we have $\utility_{\slot,i} > \utility_{\slot,i+1}$ for $i \in \{1,\ldots,\numcomponents_{\slot}-1\}$. We present in the Appendix \ref{sec:reformulation} an equivalent mixed-integer programming reformulation of~\ref{model:original} that represents components explicitly in the formulation. 

The decomposition into components discloses the discreteness of optimal solutions to the~\Problem{}. Namely, the decisions involving the components of all items (except for at most one component) are binary. Therefore, we can focus on solutions for the~$\Problem{}$ with a discrete structure when developing the approximation algorithms presented in this section. 
\begin{remark}\label{prop:partialcomponent}
    Every instance of the~$\Problem$ admits an optimal solution where at most one component~$\component^*_{\slot,i}$ is partially utilized. Moreover, $\component^*_{\slot,i}$ has the smallest utility across all utilized components. The formal proof of this property is presented in Appendix \ref{sec:reformulation}.
\end{remark}

\subsection{FPTAS for the~$\Problem{}$}
\label{sec:FPTAS1}

We design a dynamic programming algorithm that solves the~\Problem{} exactly. Although the state space of our formulation is infinitely large, it allows for the application of a discretization that delivers an FPTAS. Specifically, we show how to derive a $(1 - \epsilon)$-approximation algorithm to the~\Problem{} that runs in polynomial time in~$|\slots|$ and~$\displaystyle\frac{1}{\epsilon}$ for any given~$\epsilon > 0$.

\subsubsection{Dynamic programming algorithm}

We design an algorithm that iteratively solves sub-problems $\Problem(\slot)$ for each~$\slot$ in~$\slots$, where $\Problem(\slot)$ has the same input parameters, constraints, and objective function as~$\Problem{}$, but considers only solutions in which $\slot$ is the \textit{only} item that may have a component being \textit{partially} utilized. Remark~\ref{prop:partialcomponent} shows that every instance of the~$\Problem{}$ admits an optimal solution consisting of, at most, one partially utilized component. Therefore, we can focus our search on solutions satisfying this property.

Our algorithm is presented in Algorithm~\ref{algo:dp}, which is divided into two steps: 1) we identify solutions that fully utilize components of items in $\slots \setminus \{\slot\}$, and 2) we examine whether these solutions can be improved through the incorporation of~$\slot$.

{
\footnotesize
\begin{algorithm}
\footnotesize

    $\fullutility^* = 0$ \hfill \textit{Best objective value}

    \For{$\slot \in \slots$}{

        \For{$(i,l,\mu) \in \mathcal{M}^{\slot}$  }{  
        
            $d^{\slot}(i,l,\mu) = \knapsackcapacity + 1$ \hfill \textit{Initialization of~$d(\cdot)$}
        }

        \For{$i \in \{0,1,\ldots,|\mathcal{O}(\slot)| \}$}{
        
            $d^{\slot}(i,0,0) = 0$ \hfill \textit{Base case: solutions with zero components}
        }
    
        
        $i = 1$
        
        \For{$i' \in \{1,\ldots,|\slots|-1 \}$}{
    
        \For{$k \in \{1,\ldots,\numcomponents_{\pi(i')} \}$}{
        
            \For{$l \in \{1,\ldots,\cardinalitycapacity \}$}{
            
                \For{$\mu \in \left[0,\sum\limits_{\slot' \in \slots}
    \utilities_{\slot'}(\maxcapacity_{\slot}')
     \right]$}{
            
                    \If{$i >  \sum\limits_{i'' = 1}^{l-1} \numcomponents_{\pi(i'')}$} {
        
                        $d^{\slot}(i,l,\fullutility)  = d^{\slot}(i-1,l,\fullutility)$  \hfill         \textit{Case 1: skip element~$i$}
                    
                        \If{$\fullutility'_{i} \leq \fullutility$ \textbf{and} $d^{\slot}(i,l,\fullutility)  >                                 d^{\slot}(i-k,l-1,\fullutility - \fullutility'_{i}) + \maxcapacity'_{i}$}{
                    
                            $d^{\slot}(i,l,\fullutility)  =
                             d^{\slot}(i-k,l-1,\fullutility - \fullutility'_{i}) + \maxcapacity'_{i}$ \hfill         \textit{Case 2: include element~$i$}
                        }

                        \If{$\fullutility > \fullutility^*$ \textbf{\text{and}} $d^{\slot}(i,l,\fullutility) \leq \knapsackcapacity$}{
                        
                            $\fullutility^* = \fullutility$ \hfill  \textit{Update objective value}
                        
                        }
                    }
                }
            }
            $i = i + 1$        \hfill         \textit{Update element index}
            }
        }

        \For{$l \in \{1,\ldots,\cardinalitycapacity-1 \}$}{
                
            \For{$\mu \in \left[0,\fullutility^*\right]$}{
            
                \If{$d^{\slot}(|\mathcal{O}(\slot)|,l,\fullutility) \leq \knapsackcapacity$}{
                
                    \If{$\fullutility + 
        \utilities_{\slot}\left( \min( \knapsackcapacity - d^{\slot}(i,k,\fullutility), \quad \maxcapacity_{\slot}   )  \right)
        > \fullutility^*$}{
            $\fullutility^* = \fullutility + 
        \utilities_{\slot}\left( \min( \knapsackcapacity - d^{\slot}(i,l,\fullutility), \maxcapacity_{\slot}   )  \right)$
        \hfill
        \textit{Add~$\slot$ and update objective value}

                    }
                
                }
            
            }
            
        }
        
    }
  \caption{Dynamic programming algorithm for the~$\Problem{}$}\label{algo:dp}
\end{algorithm}
}

\noindent\textbf{Step 1:} We construct solutions containing fully-utilized components in~$\slots \setminus \{\slot\}$ for the sub-problem~$\Problem(\slot)$ in a dynamic programming approach. We consider an arbitrary permutation~$\pi:\slots \setminus \{\slot\} \rightarrow \{1,2,\ldots,|\slots|-1\}$ to build the state space, where~$\pi(k)$ denotes the~$k$-th item in~$\pi$. Recall that item $\slot$ consists of $\numcomponents_{\slot}$ components.

Given~$\pi$, we construct an ordered sequence~$\mathcal{O}(\slot)$ with~$\sum\limits_{i=1}^{|\slots|-1}\numcomponents_{\pi(i)}$ elements. 
The first~$\numcomponents_{\pi(1)}$ elements of~$\mathcal{O}(\slot)$ are associated with item~$\pi(1)$. For $i=1,\dots, \numcomponents_{\pi(1)}$, element~$o_{i}$ represents the full utilization of (only) the first~$i$ components of item~$\pi(1)$. That is, element~$o_{i}$ corresponds to the case in which the first~$i$ components of item~$\pi(1)$ are fully utilized, but other components of~$\pi(1)$ are not utilized. The weight of element~$o_{i}$ is defined as $\maxcapacity'_{i} \coloneqq \sum\limits_{i' = 1}^{i} \maxcapacity_{\pi(1),i'}$, and its utility is defined as $\fullutility'_{i} \coloneqq \utilities_{\pi(1)}\left(\sum\limits_{i' = 1}^{i} \maxcapacity_{\pi(1),i'} \cdot r_{\pi(1),i'}\right)$. The next~$\numcomponents_{\pi(2)}$ elements of~$\mathcal{O}(\slot)$ are associated with item~$\pi(2)$; for $i=1,\dots, \numcomponents_{\pi(2)}$, element~$o_{\numcomponents_{\pi(1)}+i}$ represents the full utilization of (only) the first~$i$ components of~$\pi(2)$; it has  weight $ \maxcapacity'_{\numcomponents_{\pi(1)}+i} \coloneqq \sum\limits_{i' = 1}^{i} \maxcapacity_{\pi(2),i'}$, and value  $\fullutility'_{\numcomponents_{\pi(1)}+i} \coloneqq \utilities_{\pi(2)}\left(\sum\limits_{i' = 1}^{i} \maxcapacity_{\pi(2),i'} \cdot r_{\pi(2),i'}\right)$. The following elements in $\mathcal{O}(\slot)$ are defined similarly, i.e., we construct  $\numcomponents_{\pi(\slot')}$ elements for item~$\pi(\slot')$ for each~$\slot' \in \slots \setminus \{\slot\}$.

We define~$d^{\slot}(i,l,\fullutility)$ for triples in $\mathcal{M}^{\slot} \equiv \{0, 1,\ldots, |\mathcal{O}(\slot)|\} \times \{0, 1,\ldots,\cardinalitycapacity\} \times \left[0,\sum\limits_{\slot' \in \slots}
\utilities_{\slot'}(\maxcapacity_{\slot}')
 \right]$. The value of $d^{\slot}(i,l,\fullutility)$ is the smallest amount of knapsack capacity used by a feasible solution to~$\Problem{}(\slot)$, containing~$l$ out of the first $i$ elements in~$\mathcal{O}(\slot)$ and attaining objective value~$\fullutility$. Although~$\fullutility$ may take any value in the continuous space~$\left[0,\sum\limits_{\slot \in \slots}
\utilities_{\slot'}(\maxcapacity_{\slot}')\right]$, we later apply a discretization technique to restrict the domain to a discrete set of values polynomially bounded in~$|\slots|$ and~$\frac{1}{\epsilon}$. Moreover, given the knapsack capacity~$\knapsackcapacity$, we define~$d^{\slot}(i,l,\fullutility) = \knapsackcapacity+1$ if the corresponding cases are not attainable; that is,~$\fullutility$ is too large and cannot be attained by~$l$ out of the first~$i$ items in~$\mathcal{O}(\slot)$. 
 
 According to the definition of element~$o_{i}$, a feasible solution must not include more than one element associated with the same item, which we refer to as the \textit{basic feasibility condition}. Because of this condition, if the first~$i$ elements of~$\mathcal{O}(\slot)$  are associated with fewer than~$l$ items, we cannot pick~$l$ or more elements and thus we must have $d^{\slot}(i,l,\fullutility) = \knapsackcapacity+1$ if $i \leq  \sum\limits_{i' = 1}^{l-1} \numcomponents_{\pi(i')}$.

In Algorithm~\ref{algo:dp}, we construct~$d^{\slot}(\cdot)$ iteratively. To start the algorithm, we initialize the optimal value~$\fullutility^*$ as zero and set~$d^{\slot}(i,l,\fullutility) = \knapsackcapacity+1$ for all states~$(i,l,\fullutility)$ to indicate that no feasible solution has been identified yet.  We also set~$d^{\slot}(i,0,0) = 0$ for all~$i \in \{0,\ldots,|\mathcal{O}(\slot)|\}$  to initialize the construction; these states represent feasible solutions that do not use any element of~$\mathcal{O}(\slot)$).

For other entries of~$d^{\slot}(\cdot)$, we sequentially analyze the incorporation of~$o_i$, the $i$-th element in $\mathcal{O}(\slot)$. Note that $o_i$ can be equivalently represented as the $k$-th element associated with item $\pi(i')$, that is, $i = \sum\limits_{i''=1}^{i'-1}\numcomponents_{\pi(i'')} + k$. We first check the basic feasibility condition: if $i \leq  \sum\limits_{i'' = 1}^{l-1} \numcomponents_{\pi(i'')}$, $d^{\slot}(i,l,\fullutility)$ remains as~$\knapsackcapacity+1$; otherwise, we inspect the cases of including and not including element $i$.

If we do not include~$o_i$, all~$l$ elements delivering utility~$\fullutility$ are chosen out of the first~$i-1$, and thus, we have~$d^{\slot}(i,l,\fullutility)=d^{\slot}(i-1,l,\fullutility)$. If~$\fullutility'_i \leq \fullutility$, we need to consider the case of including~$o_i$. In this case, the state~$(i,l,\fullutility)$ is attained by including element~$i$ and obtaining utility~$\fullutility - \fullutility'_i$ using~$l-1$ out of the first~$i-1$ elements of~$\mathcal{O}(\slot)$. If this solution uses a smaller capacity to reach state~$(i,l,\fullutility)$, Algorithm~\ref{algo:dp} updates the value of $d^{\slot}(i,l,\fullutility)$. We note that, according to the basic feasibility condition, element~$o_i$ can be included only if other elements associated with the same item are not included. If $o_i$ is the $k$-th element associated with item $\pi(i')$ and $o_i$ is incorporated into a solution, the update of $d^{\slot}(i,l,\fullutility)$ is based on $d^{\slot}(i-k,l-1,\fullutility-\fullutility_i')$. After updating the value of~$d^{\slot}(i,l,\fullutility)$, we update the optimal solution if the obtained solution is feasible (i.e., $d^{\slot}(i,l,\fullutility) \leq \knapsackcapacity$) and better (i.e.$\fullutility>\fullutility^*$). 

At the end of Step 1, the entries of~$d^{\slot}(\cdot)$ represent all the solutions to~$\Problem{}$ that consist of fully-utilized components associated with at most~$\cardinalitycapacity$ items in~$\slots \setminus \{\slot\}$.

\noindent\textbf{Step 2:} 
Next, we examine whether we can construct new solutions based on~$d^{\slot}(i,l,\fullutility)$ with $l\in\{0,1,\ldots,\cardinalitycapacity-1\}$ by additionally utilizing components of~$\slot$ and using the remaining space~$\knapsackcapacity - d^{\slot}(i,l,\fullutility)$. Note that components of~$\slot$ can be partially utilized, and since $l < \cardinalitycapacity$, the incorporation of components of~$\slot$ does not violate the cardinality capacity. Therefore, these new solutions are feasible to the~$\Problem$.

\subsubsection{Approximation scheme based on state space discretization:} For each~$\slot$ in~$\slots$, the complexity of Algorithm~\ref{algo:dp} depends on the size of the state space~$\mathcal{M}^{\slot}$. Particularly, the third dimension may assume any value in the continuous space~$\left[0,\sum\limits_{\slot \in \slots}
\utilities_{\slot'}(\maxcapacity_{\slot}')
\right]$.
Moreover, as the~\Problem{} is NP-hard, each subproblem~$\Problem(\slot)$ is also NP-hard, so the problem cannot be solved to optimality in polynomial time unless $P = NP$. Instead, we propose an approximation scheme that executes
Algorithm~\ref{algo:dp} in polynomial time by replacing~$\mathcal{M}^{\slot}$ with a polynomially-bounded discretized space. The discretization makes the execution times of both  Steps 1 and 2 polynomially bounded. Additionally, it only results in a precision loss in Step 1. Thus, we only need to control the errors in the construction of~$d^{\slot}(i,l,u)$ for~$\slot \in \slots$. 

We adapt an FPTAS for the knapsack problem to solve the~$\Problem(\slot)$, see, e.g., \citep{ibarra1975fast,vazirani2013approximation,caprara2000approximation}. These approximation schemes discretize the third dimension of~$\mathcal{M}^{\slot}$ by replacing the value~$\fullutility'_i$ of each element~$o_i$ in~$\mathcal{O}(\slot)$ by the value of~$k \Delta$ such that $k \Delta \leq  \fullutility'_i < (k+1) \Delta$, $k \in \mathbb{N}$, and~$\Delta \in O\left(\frac{|\mathcal{O}(\slot)|^c}{e^{c'}}\right)$, where~$c$ and~$c'$ are constant values that do not depend on the dimensions of the problem. As a result, $\mathcal{M}^{\slot}$ is replaced by a discrete state space with a size that is polynomially bounded in~$n$ and~$\frac{1}{\epsilon}$, which allows us to compute~$\underline{d}^{\slot}(\cdot)$ as counterparts to $d^{\slot}(\cdot)$ using Algorithm~\ref{algo:dp}. The optimal solution to the discretized problem is a $(1 - \epsilon)$-approximation to the original~$\Problem(\slot)$. Finally, incorporating~$\slot$ to each solution in Step 2 does not incur losses in the objective function and thus does not worsen the approximation factor. Therefore, by combining 
Algorithm~\ref{algo:dp}  and an FPTAS for~$\Problem(\slot)$, we obtain an FPTAS for the $\Problem$ with running time~$O(|\slots| \frac{m^2}{\epsilon} \cardinalitycapacity)$, where~$m \coloneqq \sum\limits_{\slot \in \slots}\numcomponents_{\slot}$ is the number of components composing all items in~$\slots$.

\begin{theorem}\label{thm:selectionfptas}
    The $\Problem$ admits an FPTAS with running time~$O(|\slots| \frac{m^2}{\epsilon} \cardinalitycapacity)$, where~$m \coloneqq \sum\limits_{\slot \in \slots}\numcomponents_{\slot}$.
\end{theorem}

\subsection{Greedy Algorithm for the~$\Problem{}$}\label{sec:submodularity}

Algorithm~\ref{alg:greedy} is a greedy approach to solve the~\Problem{}. This algorithm is based on iteratively solving the \textit{$\Problem{}$ with no cardinality constraint}, which we denote by~$\Problem{}^T$. Differently from the~\Problem{}, the~$\Problem{}^T$ can be solved to optimality by \citet{dantzig1957discrete}'s greedy algorithm, as it is equivalent to the continuous knapsack problem over the~$m$ components.

\begin{proposition}\label{prop:greedyoptimal}
The~$\Problem{}^T$  can be solved in time~$O(m \log{m})$.
\end{proposition}

For~$\slots' \subseteq \slots$, we can construct an instance~$I(\slots')$ of the~$\Problem{}^T$ with knapsack constraint $\knapsackcapacity$ and no cardinality constraint over the set of items~$\slots'$. We denote $G(\slots'): 2^{\slots} \rightarrow \mathbb{R}$ as the optimal objective value of~$I(\slots')$. Algorithm \ref{alg:greedy} starts with an empty set~$\slots'$; in each step, it adds an item $\slot^*$ that leads to the maximum increase in $G(\cdot)$, that is, $\slot^*=\argmax\limits_{\slot \in \slots \setminus \slots'} G(\slots' \cup \{\slot\})$. After~$\cardinalitycapacity$ iterations, $\slots'$ contains~$\cardinalitycapacity$ elements, and Algorithm~\ref{alg:greedy} returns a solution to the original instance of the~$\Problem{}$.

{
\footnotesize
\begin{algorithm}[ht!]
\footnotesize
 $\slots' = \emptyset$

\For{$k = 1$ to $\cardinalitycapacity$}{
    
     $\slot^* = \argmax\limits_{\slot \in \slots \setminus \slots'} G(\slots' \cup \{\slot\})$
     
    $\slots' = \slots' \cup \{\slot^*\}$
    }

 return $G(\slots')$
\caption{Greedy algorithm for the~\Problem{}}
\label{alg:greedy}
\end{algorithm}
}
Each step of Algorithm~\ref{alg:greedy} requires the solution of~$O(|\slots|)$ instances of the~$\Problem{}^T$. The solution of each sub-problem~$\Problem{}^T$ can be identified in time~$O(m \log{m})$ (see Proposition~\ref{prop:greedyoptimal}). Therefore, Algorithm~\ref{alg:greedy} runs in time~$O(\cardinalitycapacity|\slots|m\log{m})$.

Proposition~\ref{prop:pointwisesubmodular} shows that  function~$G(\cdot)$ is monotone submodular.
\begin{proposition}\label{prop:pointwisesubmodular} 
$G(\cdot)$ is monotone submodular.
\end{proposition}

It follows from Proposition~\ref{prop:pointwisesubmodular} and the result by \citet{nemhauser1978analysis} that Algorithm~\ref{alg:greedy} is an algorithm with a constant-factor approximation guarantee.

\begin{theorem}\label{thm:greedy} Algorithm~\ref{alg:greedy} is a~$\left(1 - \frac{1}{e}\right)$-approximation for the~$\Problem$.
\end{theorem}
\begin{proof}{Proof of Theorem~\ref{thm:greedy}:} For any set~$S$  and for any monotone submodular set function~$f:2^{S} \rightarrow \mathbb{R}$ defined over~$S$, \citet{nemhauser1978analysis} show that a greedy algorithm gives a $(1-1/e)$-approximation for the problem $\max\limits_{S' \subseteq S: |S'| = k}f(S')$, $0 \leq k \leq |S|$. From Proposition~\ref{prop:pointwisesubmodular}, $G(\cdot)$ is monotone submodular, so it follows that Algorithm~\ref{alg:greedy} is a~$\left(1 - \frac{1}{e}\right)$-approximation for the~$\Problem$. 
\hfill$\blacksquare$
\end{proof}

Finally, \citet{conforti1984submodular} show that the greedy algorithm is actually a~$(\frac{1 - e^{-\gamma}}{\gamma})$-approximation for the maximization of submodular functions with a cardinality constraint,
where
\[
\gamma = \max_{\slot \in \slots}\frac{G(\{\slot\}) - G(\slots) + G(\slots \setminus \{\slot\})}{G(\{\slot\})}.
\]
The value of~$\gamma$ can be interpreted as an indicator of how far the function is from being additive. We have $\gamma = 1$ if there is an item~$\slot$ not being used in the optimal solution. In our case, $\gamma = 1$ is guaranteed to hold because of the cardinality constraint; thus, the approximation guarantee of Theorem~\ref{thm:greedy} is tight. Nevertheless, we note that the results of our numerical evaluation of Algorithm~\ref{alg:greedy} suggest that the empirical approximation ratio is better than the theoretical guarantee of 0.63 ensured by Theorem~\ref{thm:greedy}.

\begin{remark}\label{remark:general utilities}
Extending the~\Problem{} to more general utility functions is nontrivial, as optimal item utilizations may be irrational numbers that do not admit algebraic representations and, therefore, are not finitely representable (see~\citet{hochbaum1995nonlinear}). In particular, if the utility functions are not necessarily piecewise linear, the evaluation of the objective function becomes challenging, as it 
generalizes the square-root sum problem, which is as follows: given a set of integer numbers~$d_1, d_2, \ldots, d_n$ and an integer number~$k$, decide whether~$\sum\limits_{i = 1}^{n}\sqrt{d_i} \leq k$. To the best of the author's knowledge, the strongest result to date shows that the square-root sum problem belongs to the counting hierarchy~(\cite{allender2009complexity}), and it is still not clear whether the square-root sum problem belongs even to NP. However, we can derive an FPTAS for the extensions of the~$\Problem{}$ where the utility functions are continuous, injective, and Lipschitz continuous, with a dynamic programming formulation similar to the one used in Algorithm~\ref{algo:dp}. We present the details in the Appendix~\ref{sec:Continuous Utilities}.
\end{remark}

\section{Online~\Problem{} in the Random Order Model}\label{sec:online}

Following the definition of offline~\Problem{} in Section \ref{sec:definition}, we investigate an online version of the~\Problem{} where the items arrive sequentially, the utility and capacity of an item are only revealed upon its arrival, and the decision on the amount taken from each item is irrevocable and must be made before the next item arrives. In settings where an adversarial environment fully controls the items' utilities, capacities, and arrival sequence, the knapsack problem admits no~$c$-competitive algorithm for any constant~$c$~\citep{marchetti1995stochastic}; this result extends to the~\Problem{}. Therefore, we investigate the online~\Problem{} in the \textit{random order model}, in which an adversary may set the utility and capacity of the items, but the number of arriving items is known a priori, and the arrival order is uniformly distributed over the set of all permutations. Thus, the competitive ratio is defined over the \textit{expected} objective value of the solutions selected by the online algorithm, following Definition~\ref{def:ratio}.  

Our analysis and algorithm draw on the findings of~\citet{giliberti2021improved} for the continuous knapsack problem
and~\citet{kesselheim2018primal} and~\citet{albers2021improved} for the 0-1 knapsack problem and the general assignment problem (an extension of the knapsack problem where items may be assigned to several different knapsacks), all in the random order model. We discuss the commonalities and differences with these papers in Section~\ref{sec:relatedwork}.

\subsection{Additional Notation and Assumptions}\label{sec:add_notation}
In the random order model, the arrival order of the items is uniformly sampled out of the~$n!$ permutations of~$\slots$. We use~$\slots_l$ to denote the first~$l$ incoming items of~$\slots$ and~$\slots_{l,k} \equiv  \slots_{k} \setminus \slots_{l}$ is the subsequence~$\slot_l, \slot_{l+1},\ldots,\slot_k$, where $\slot_l$ refers to the item arriving at position $l$. We do not include explicit references to the specific realization of the arrival order in the notation, as they are unnecessary for our analysis. Finally, we assume that~$\cardinalitycapacity \geq 2$, as the online~$\Problem$ is equivalent to the secretary problem if~$\cardinalitycapacity = 1$ and therefore admits an $e$-competitive online algorithm~(\cite{ferguson1989solved}).

We use~$\Problem{}(\slots)$ to denote an instance of the~\Problem{} defined over the set of items~$\slots$. For every~$l$ in $\{1,2,\ldots,\numitems\}$, let $s^{(l)} = (\vecx^{(l)},\vecy^{(l)})$ denote a solution produced by a tractable~$\alpha$-approximation algorithm~$\Approx_{\alpha}$ to the~$\Problem{}(\slots_l)$ for some~$\alpha \in (0,1]$, and $x_{\slot_k}^{(l)}$ and $y_{\slot_k}^{(l)}$  denotes the utilization of the knapsack and carnality capacity by $\slot_k$ (the item arriving at position $k$) in~$s^{(l)}$. Our algorithm sets the values of~$x_{\slot_l}$ and~$y_{\slot_l}$ to determine the utilization of~$\slot_l$
based on~$x_{\slot_l}^{(l)}$ and $y_{\slot_l}^{(l)}$.

\subsection{Description of the Algorithm} 
We propose Algorithm \ref{alg:online} (denoted by~$\online$) to solve the online~\Problem{} in the random order model. The algorithm is iterative, and each step~$l$ in $\{1,2,\ldots,\numitems\}$ is associated with the arrival of an item~$\slot_l$. Algorithm~\ref{alg:online} stores the solution using~$\slots_{\online}$, which contains the set of incorporated items, and the vectors~$\vecx$ and~$\vecy$, where~$x_{\slot}$ and~$y_{\slot}$ are the knapsack and cardinality capacities allotted to~$\slot$, respectively. Moreover, $\online$ keeps track of the remaining (residual) knapsack and cardinality capacities using variables ~$\knapsackcapacity_{\online}$ and~$\cardinalitycapacity_{\online}$, respectively.

{
\footnotesize
\begin{algorithm}[ht!]
\footnotesize
    $\slots_{\online} = \emptyset$, 
    $\knapsackcapacity_{\online} = \knapsackcapacity$, 
    $\cardinalitycapacity_{\online} = \cardinalitycapacity$,
    $r^* = 0$, $\vecx = \textbf{0}$, $\vecy = \textbf{0}$
    
    \For{$l = 1$ to $c \cdot \numitems$}{ 

        $r^* = \max(r^*,R_{\slot_l}\left(\maxcapacity_{\slot_l})\right)$
    }
    \For{$l = c \cdot \numitems+1$ to $d \cdot \numitems$}{
    
        \If{$R_{\slot_l}\left(\maxcapacity_{\slot_l}\right) > r^*$ 
        \textbf{and} $\knapsackcapacity_{\online} > 0$ 
        \textbf{and} $\cardinalitycapacity_{\online} = \cardinalitycapacity$ 
        }{
        
            $\slots_{\online} = \slots_{\online} \cup \{\slot_l\}$,
            $\knapsackcapacity_{\online} = \knapsackcapacity_{\online} - \min(\maxcapacity_{\slot_l},\knapsackcapacity_{\online}) $,
            $\cardinalitycapacity_{\online} = \cardinalitycapacity_{\online} - 1$,
            $x_{\slot_l} = \maxcapacity_{\slot_l}$,
            $y_{\slot_l} = 1$
        }
    }
    \For{$l = d \cdot \numitems+1$ to $\numitems$}{
    
        $(\vecx^{(l)},\vecy^{(l)}) = \Approx{}_{\alpha}(\slots_l)$ 
    
        \If{$x^{(l)}_{\slot_l} > 0$ \textbf{and} $\knapsackcapacity_{\online} > 0$ \textbf{and} $\cardinalitycapacity_{\online} \geq 1$ 
        }{
            $x_{\slot_l} = \min(\beta x^{(l)}_{\slot_l},\knapsackcapacity_{\online})$,          $y_{\slot_l} = y^{(l)}_{\slot_l}$
        
            $\slots_{\online} = \slots_{\online} \cup \{\slot_l\}$,
            $\knapsackcapacity_{\online} = \knapsackcapacity_{\online} - x_{\slot_l}$,
            $\cardinalitycapacity_{\online} = \cardinalitycapacity_{\online} - 1$
        }
    }
    return $(\vecx,\vecy)$
\caption{Algorithm~$\online$ for the online~\Problem{} in the random order model}
\label{alg:online}
\end{algorithm}
}
Algorithm \ref{alg:online} consists of three phases, defined by parameters~$c$ and~$d$, $0 \leq c \leq d \leq 1$. Another parameter~$\beta \in (0,1)$ is used by~$\online$ when incorporating items in the third phase (see below). We derive values for~$c, d,$ and~$\beta$ that maximize the competitive ratio of~$\online$ at the end of the proof. 

\paragraph{Sampling phase (lines 2-3):}  In this phase, $\online$ observes the total utility~$R_{\slot_l}\left(\maxcapacity_{\slot_l}\right)$ of each item~$\slot_l$ in~$\slots_{c\cdot\numitems}$ and stores the maximum among these values in~$r^*$, without adding any of these items to~$\slots_{\online}$. Therefore, the residual capacities are equal to the original capacities (i.e., the knapsack is empty) at the end of the sampling phase.

\paragraph{Secretary phase (lines 4-6):} After~$c \cdot \numitems$ steps, Algorithm~\ref{alg:online} switches to the secretary phase, in which we use~$\online_S$ to denote the behavior of~$\online$. The goal of~$\online_S$ is to pick one item with large total utility using~$\utility^*$ as the reference value. More precisely, $\online_S$ incorporates the whole of the \textit{very first} item~$\slot_l \in \slots_{c\numitems+1,d\numitems}$ arriving during the secretary phase such that~$R_{\slot_l}(\maxcapacity_{\slot_l}) > r^*$. Recall that we assume w.l.o.g. that~$\maxcapacity_{\slot} \leq \knapsackcapacity$ for every item~$\slot$, and as the knapsack is empty when the secretary phase starts, $\online_S$ can always fully incorporate~$\slot_l$. Algorithm~$\online_S$ takes at most one item, as it only executes line 6 if $\cardinalitycapacity_{\online} = \cardinalitycapacity$  (i.e., when the knapsack is empty). If all items arriving during the secretary phase have total utility smaller than~$r^*$, the knapsack is empty at the end of step~$d \cdot \numitems$.

\paragraph{Knapsack phase (lines 7-11):} The~$(d \cdot \numitems+1)$-th arrival defines the beginning of the knapsack phase, in which we use~$\online_K$ to denote the behavior of Algorithm~\ref{alg:online}. Upon the~$l$-th arrival, $\online_K$ obtains a solution~$s^{(l)} = (\vecx^{(l)},\vecy^{(l)})$ for the~$\Problem{}(\slots_{l})$ (i.e., the~$\Problem{}$ with set of items restricted to~$\slots_l$) by solving~$\Approx_{\alpha}(\slots_{l})$. The only condition we need for our analysis is that the objective value of~$s^{(l)}$ attains a fraction not smaller than~$\alpha$ from the optimal solution; if~$\Problem{}(\slots_{l})$ admits more than one solution that satisfies this condition, any arbitrary solution may be chosen. If there are still residual knapsack and cardinality capacities upon the arrival of item~$\slot_l$ (i.e., $\knapsackcapacity_{\online} > 0$ and~$\cardinalitycapacity_{\online} \geq 1$), and~$\slot_{l}$ is utilized in~$s^{(l)}$ (i.e., if $x_{\slot_l}^{(l)} > 0$), $\online_K$ incorporates~$\slot_l$. The assigned capacity~$x_{\slot_l}$ to~$\slot_l$ is upper-bounded by the minimum between the residual capacity~$\knapsackcapacity_{\online}$ and~$\beta x_{\slot}^{(l)}$ for some~$\beta \in (0,1]$, i.e., $\online_{}$ scales the utilization in~$s^{(l)}$ when incorporating~$\slot_l$. Consequently, no item selected by~$\online_{K}$ can occupy more than a fraction $\beta$ of the knapsack capacity $\knapsackcapacity$. 

In the next sections, we study the expected utilities collected by Algorithm~\ref{alg:online}. Specifically, we derive lower bounds for the expected utilities collected in the Secretary and Knapsack phases.

\subsection{Probabilities of Item Incorporation}
\label{sec:expectations}

In this section, we investigate the probabilities with which Algorithm~\ref{alg:online} incorporates specific items during the secretary and knapsack phases. These results enable us to derive lower bounds for the expected utility collection. 

\subsubsection{Secretary phase}
\label{sec:secret}
Algorithm~$\online_{S}$ incorporates the first incoming item in~$\slots_{c \cdot \numitems+1,d \cdot \numitems }$ whose total utility is larger than~$\utility^*$. To derive a lower bound for the expected utility collected in~$\online_{S}$, we focus on the case where the item with largest total utility $\overline{\slot}$, i.e.,~$\overline{\slot} \equiv \argmax\limits_{\slot \in \slots} \utilities_{\slot}(\maxcapacity_{\slot})$, is picked. Lemma~\ref{lemma:sec} presents a lower bound for~$\overline{p}$, the probability with which~$\online_{S}$ incorporates~$\overline{\slot}$, which is asymptotic in the number of items~$\numitems$.

\begin{lemma}[Lemma 1 in \cite{albers2021improved}]\label{lemma:sec} 
The probability~$\overline{p}$ that item~$\overline{\slot}$ is incorporated in the secretary phase satisfies
$
  \overline{p} \geq c \ln{\frac{d}{c}} - o(1). 
$
\end{lemma}

\subsubsection{Knapsack phase}
\label{sec:knapsack}

Next, to derive a lower bound for the expected utility collected in~$\online_{K}$, we focus on the case where the residual capacities are equal to the original capacities at the beginning of the knapsack phase (i.e., when~$\online_{S}$ does not incorporate any item). We define the Bernoulli variable~$\xi$ to indicate the occurrence of this event, so all the results related to the performance of~$\online_{K}$ are conditional on~$\xi = 1$.

\citet{albers2021improved} characterizes the probability of event~$\xi$, which we reproduce in Lemma~\ref{lemma:cdbound}:
\begin{lemma}[Lemma 7 in~\cite{albers2021improved}]\label{lemma:cdbound} 
The residual capacities are equal to the original capacities at the beginning of the knapsack phase with probability $Pr[\xi = 1] = \frac{c}{d}$.
\end{lemma}

We use~$\knapsackcapacity_{l}$ and~$\cardinalitycapacity_{l}$ to denote the volume of knapsack and cardinality capacity that have been already consumed upon the arrival of item~$\slot_l$. In a slight abuse of notation, we use~$\knapsackcapacity_{l}$ and~$\cardinalitycapacity_{l}$ to refer to deterministic and realized values (e.g., in Algorithm~\ref{alg:online}) and as random variables; the correct interpretation will be clear from the context.

In addition to~$\xi$, we define two Bernoulli random variables for the knapsack phase to analyze the performance of~$\online_{K}$. First, we define $\phi_l=1$ to indicate that the condition $\knapsackcapacity_{l} < (1 - \beta) \knapsackcapacity$ and~$\cardinalitycapacity_{l} < \cardinalitycapacity$ is satisfied upon the arrival of the~$l$-th item, where~$\beta$ is the parameter of~$\online$ used in line 10 of Algorithm~\ref{alg:online}. As we explain later, $\phi_l$ indicates whether the residual knapsack and cardinality capacities upon the~$l$-th arrival are sufficiently large and allow the incorporation of item~$\slot_l$. Second, we use $\overline{\Phi}_{l,\slot}=1$ to indicate that item~$\slot\in \slots$ arrives at the~$l$-th position.

In Lemma~\ref{lemma:lb_l}, we derive a lower bound for~$Pr(\phi_l=1| \xi=1)$ for~$l \geq dn+1$. It is worth noting that when $\knapsackcapacity_{l} < (1 - \beta) \knapsackcapacity$ and~$\cardinalitycapacity_{l} < \cardinalitycapacity$, the residual knapsack capacity at the beginning of step~$l$ is at least $\beta \knapsackcapacity$. Therefore, in this case, $\online_{K}$ can always use $\beta x^{(l)}_{\slot_l}$ capacity from $\slot_l$ in line 10 (recall that $\slot_l$ is the item arriving at the~$l$-th position),  as $x^{(l)}_{\slot_l}\leq \maxcapacity_{\slot_l} \leq \knapsackcapacity$ implies $\beta x^{(l)}_{\slot_l}\leq\beta\knapsackcapacity$. However, $\online_{K}$ may also utilize $\beta x^{(l)}_{\slot_l}$ from item $\slot_l$ even when~$\phi_l=0$, specifically when $(1 - \beta) \knapsackcapacity < \knapsackcapacity_{l} \leq \knapsackcapacity - \beta x^{(l)}_{\slot_l}$and~$\cardinalitycapacity_{l} < \cardinalitycapacity$. Therefore, Lemma~\ref{lemma:lb_l} provides a lower bound for the probability with which $\online_{K}$ could utilize $\beta x^{(l)}_{\slot_l}$ knapsack capacity with~$\slot_l$ under the assumption that the knapsack is empty at the beginning of the knapsack phase.

\begin{lemma}\label{lemma:lb_l}
    For every~$l \geq dn+1$, we have
    \begin{equation}\label{eq:lb_l}
         Pr\left[ \phi_l  = 1 | \xi =1  \right]
     \geq  1 - \left( \frac{1}{1 - \beta} \right)\ln{\frac{l-1}{dn}}.   
        \end{equation}
    
\end{lemma}

Lemma~\ref{lemma:lb_knapsack} uses Eq.\eqref{eq:lb_l} to derive a lower bound expression for the probability with which~$\online_{K}$ incorporates an item~$\slot$ during the knapsack phase.

\begin{lemma}\label{lemma:lb_knapsack} For every~$\slot$ in~$\slots$,
  \begin{equation*}
  \sum_{l = d \numitems+1}^{\numitems} Pr[\overline{\Phi}_{l,\slot} = 1 | \xi = 1]
        Pr[\phi_{l} = 1 | \xi = 1]\geq
(1 - d)\left(\frac{2 - \beta}{1 - \beta}  \right) + \left( \frac{1}{1 - \beta} \right) \ln{d}.
        \end{equation*}
\end{lemma}

\subsection{Competitive Ratio}\label{sec:competitiveratio}

We finalize the analysis by classifying instances of the~\Problem{} into two categories based on the number of packed items in the optimal offline solutions. Thus, the competitive ratio of~$\online$ is given by the worse (smaller) expected competitive ratio attained by~$\online$ for the two categories. 

\subsubsection{Optimal offline solutions with just one item}
Let~$\Psi_1$ indicate instances of this type. As we assume~$\maxcapacity_{\slot} \leq \knapsackcapacity$ for every item~$\slot$, the optimal solution to a~$\Psi_1$ instance must fully pack the item~$\overline{\slot}$ with the largest total utility. Also, we must have $x^*_{\overline{\slot}} = \maxcapacity_{\overline{\slot}} = \knapsackcapacity$ and~$\utilities^* = \utilities_{\overline{\slot}}(\knapsackcapacity)$; otherwise, the remaining knapsack capacity could be used to increase $\utilities^*$ (note that $\cardinalitycapacity \geq 2$), which contradicts the assumption that an optimal solution uses just one item. Proposition~\ref{prop:ratio_case1} derives a lower bound for the expected utility gain $\Ex^{\Psi_1}_{\online}[\utilities(\vecx)]$ for instances of this type.

\begin{proposition}
    \label{prop:ratio_case1}
The expected utility gain $\Ex^{\Psi_1}_{\online}[\utilities(\vecx)]$ for instances of type~$\Psi_1$ is
\begin{equation}\label{eq:ratio_case1}
    \Ex^{\Psi_1}_{\online}[\utilities(\vecx)] \geq
        \utilities^*\left[
     c \ln{\frac{d}{c}} + \frac{c}{d}\beta  \left((1 - d)\left(\frac{2 - \beta}{1 - \beta}  \right) + \left( \frac{1}{1 - \beta} \right) \ln{d}  \right)- o(1)  
    \right].
\end{equation}
\end{proposition}

\subsubsection{Optimal offline solutions with more than one item}
Let~$\Psi_2$ indicate this category of instances. We derive a lower bound for the utility gain for such instances in Proposition~\ref{prop:profit_knapsack}.

\begin{proposition}
\label{prop:profit_knapsack}
The expected utility gain $\Ex^{\Psi_2}_{\online}[\utilities(\vecx)]$ for instances of type~$\Psi_2$ is
\begin{equation}\label{eq:ratio_case3}
        \Ex^{\Psi_2}_{\online}[\utilities(\vecx)] 
    \geq 
        \alpha \utilities^*\left[\frac{c}{\cardinalitycapacity} \ln{\frac{d}{c}} + \frac{c}{d}\beta  \left((1 - d)\left(\frac{2 - \beta}{1 - \beta}  \right) + \left( \frac{1}{1 - \beta} \right) \ln{d}  \right)- o(1)
    \right].
\end{equation}
\end{proposition}

\subsubsection{Competitive factor}\label{sec:competitive_factor}

We note that the lower bound in Eq.~\eqref{eq:ratio_case3} is strictly smaller than the lower bounds in Eq.~\eqref{eq:ratio_case1} for any choice of~$\beta$, $c$, and $d$, and the competitive factor~$\alpha$ is fixed for the selected algorithm~$\Approx_{\alpha}$. Also, $o(1)$ is asymptotic regarding the number of items $n$ and $n \rightarrow \infty$ in the worst-case scenario. Therefore, for a given value of $\cardinalitycapacity$, we can obtain the values of~$c$, $d$, and~$\beta$ that maximize the asymptotic competitive ratio of~$\online$ by solving the following problem:
\begin{equation}
\label{eq:perf_obj}
\max\limits_{0 \leq c \leq d \leq 1, 0 < \beta < 1} \quad 
\frac{c}{\cardinalitycapacity} \ln{\frac{d}{c}} + \frac{c}{d}\beta  \left((1 - d)\left(\frac{2 - \beta}{1 - \beta}  \right) + \left( \frac{1}{1 - \beta} \right) \ln{d}  \right).
\end{equation}

We have two remarks about Eq.\eqref{eq:perf_obj}. First, the value of Eq.\eqref{eq:perf_obj} decreases with the increase of $\cardinalitycapacity$. Therefore, we can derive an asymptotic competitive ratio that applies to all values of $\cardinalitycapacity$ by studying the case where $\cardinalitycapacity \rightarrow \infty$. Second, Eq.\eqref{eq:perf_obj} is a complex nonlinear function of parameters $c$, $d$, and~$\beta$, the global maxima of which is hard to identify. However, $\Approx_{\alpha}$ parameterized by any feasible values of $c$, $d$, and~$\beta$ provides a lower bound for its true performance. Therefore, we solve for local maxima of Eq.\eqref{eq:perf_obj} with $\cardinalitycapacity \rightarrow \infty$, based on which we derive the asymptotic competitive ratio applicable to all values of $\cardinalitycapacity$ presented in Theorem~\ref{thm:ratio}.

\begin{theorem}\label{thm:ratio}
    Algorithm \ref{alg:online} is $\frac{10.427}{\alpha}$-competitive for the online~\Problem{} in the random order model if parameterized by $c  = d \approx 0.695$ and~$\beta \approx 0.560$. 
\end{theorem}

We note that when we consider the case~$\cardinalitycapacity \rightarrow \infty$, the first term of Eq.\eqref{eq:perf_obj} vanishes, which essentially deems the secretary phase inconsequential in deriving the general competitive ratio for all values of~$\cardinalitycapacity$  in Theorem~\ref{thm:ratio}. This property is also reflected in the choice of parameters~$c = d$, which eliminates the secretary phase. In addition to the general competitive ratio in Theorem~\ref{thm:ratio}, we can derive competitive ratios specific for a given value of $\cardinalitycapacity$, which enables us to improve our evaluation of the performance of Algorithm \ref{alg:online}. When the value of $\cardinalitycapacity$ is small,  Corollary~\ref{cor:ratioC2} exemplifies the improvement in the performance bound. 

\begin{corollary}\label{cor:ratioC2} Algorithm \ref{alg:online} is $\frac{5.295}{\alpha}$-competitive for the online~\Problem{} in the random order model  parameterized by $c  = 0.3775$, $d = 0.915$, and~$\beta = 0.79$ if~$\cardinalitycapacity = 2$.
\end{corollary}

We also show the improvement in the competitive ratio in the case where the cardinality constraint is trivially satisfied in the knapsack phase, i.e., when~$\cardinalitycapacity > (1-d)\numitems$.

\begin{corollary}\label{cor:ratio2}
    Algorithm \ref{alg:online} is $\frac{6.401}{\alpha}$-competitive for the online~\Problem{}  
    in the random order model 
    if parameterized by $c  = d = \beta \approx 0.431$  in the special case where $\cardinalitycapacity \geq 0.569 \numitems$. 
\end{corollary}

The~$\alpha$ in our analytical results stands for the approximation guarantee of the approximation algorithm to the~\Problem{} used in Algorithm~\ref{alg:online}. Theorem \ref{thm:greedy} shows that $\alpha=1-\frac{1}{e}$ if we embed Algorithm~\ref{alg:greedy} into Algorithm~\ref{alg:online}. More generally, Theorem~\ref{thm:selectionfptas} shows how to construct an efficient~$\alpha$-approximation algorithm for any~$\alpha < 1.0$.

\begin{remark}
Algorithm \ref{alg:online} and the results in this section extend to generalizations of the~\Problem{} whose offline version admits efficient algorithms with approximation guarantees (see Remark~\ref{remark:general utilities}); in particular, the results hold when all utility functions are concave, continuous but not necessarily piecewise linear, and smooth (e.g., Lipschitz continuous). 
\end{remark}
\begin{remark}
The incorporation of~$\beta$ allows us to improve the performance guarantees of Algorithm \ref{alg:online}. However, the approximation guarantees of~$\online$ do not extend to the 0-1 knapsack version of the~$\Problem{}$, in which~$\beta$ must be equal to one, and thus, our approximation guarantees are not defined. Although our analysis can be modified to handle this case, the theoretical guarantees are not as strong as~\citet{albers2021improved}.
\end{remark}

\section{Numerical Studies}\label{sec:experiments}

We investigate the empirical performance of Algorithms~\ref{alg:greedy} (the greedy algorithm) and~\ref{alg:online}
(the online algorithm). We implement both in Python 3.10.4, and use Pyomo with Gurobi 10.0.2 to solve the mixed-integer linear programming formulations~\citep{hart2011pyomo,bynum2021pyomo,gurobi}. The experiments are executed on an Apple  M1 Pro with 32 GB of RAM. 

We created test instances based on experiments in~\citet{de2003polyhedral}. Specifically, instances are constructed for the following component-based formulation~\ref{model:Knapsack}, in which variables~$x_{\slot,i}\in [0,1]$ represent the utilized percentage of component~$\component_{\slot,i}$; and $a_{\slot,i}$ and $r_{\slot,i}$ are the weight and utility if selecting the whole component~$\component_{\slot,i}$. 
\[
\begin{array}{ccll}\label{model:Knapsack}
\tag{\textbf{\texttt{CKP-K}}}
 &\max
& \sum\limits_{\slot \in \slots} \sum\limits_{i=1}^{\numcomponents_{\slot}} \utility_{\slot,i} \cdot x_{\slot,i}  
\\
&(a) & \sum\limits_{\slot \in \slots} \sum\limits_{i=1}^{\numcomponents_{\slot}}  a_{\slot,i} \cdot x_{\slot,i}   \leq \knapsackcapacity \\
&(b)& \sum\limits_{\slot \in \slots} y_{\slot} \leq \cardinalitycapacity  & \\
&(c)&x_{\slot,i} \leq y_{\slot} &\forall (\slot,i) \in \slots \times \{1,2,\ldots,\numcomponents_{\slot}\} \\
&& y_{\slot}\in \{0,1\} & \forall \slot \in \slots \\
&& x_{\slot,i} \in [0,1] & \forall (\slot, i)\in \slots \times \{1,2,\ldots,\numcomponents_{\slot}\}
\end{array}
\]
The values of~$r_{\slot,i}$ and~$a_{\slot,i}$ are generated such that the utility-to-weight ratios $\frac{r_{\slot,i}}{ a_{\slot,i}}$ in each item are concave. We provide more detail about our instance generation process in the following subsections.

\subsection{Performance of the Greedy Algorithm}\label{sec:experiments_greedy}
 To test the performance of Algorithm \ref{alg:greedy}, we construct dataset \texttt{A} as follows. We generate ten instances for each value of~$|\slots|$ in $\{10,20,30,\ldots, 90, 100, 250, 500\}$ and~$\cardinalitycapacity$ in $\{2,\lfloor 0.3|\slots| \rfloor,\lfloor 60.6|\slots| \rfloor\}$, yielding a total of 360 instances. Within the test instances, each item~$\slot$ consists of two components, i.e., its piecewise linear utility function consists of two line segments. For each item, we first sample independently and uniformly at random two values, denoted as $r$ and~$r'$, from $[10, 25]$ and two values, denoted as $a$ and~$a'$, from $[5,20]$. The first component of~$\slot$ has utility $\utility_{\slot,1} = \max(r,r')$ and weight  $a_{\slot,1} =  \min(a,a')$, whereas the second component has utility $\utility_{\slot,2} = \min(r,r')$ and weight $\maxcapacity_{\slot,2} = \max(a,a')$; thus, the utility-to-weight structure of~$\slot$  is concave piecewise linear by construction. 
The knapsack capacity is selected as 
$W \equiv \max\left(0.3 \sum\limits_{\slot \in \slots} (a_{\slot,1}+a_{\slot,2}), 1+\max\limits_{\slot \in \slots}(a_{\slot,1}+a_{\slot,2})\right)$.  

Figure~\ref{fig:greedy} shows the empirical performance of Algorithm \ref{alg:greedy} for the instances of dataset~\texttt{A}. The results are aggregated by the number of items~$|\slots|$ , indicated in the~$x$ axis of the plot, and each curve represents a cardinality capacity. The~$y$ axis presents the empirical approximation ratio attained by the greedy algorithm, which is the ratio between the objective value achieved by Algorithm \ref{alg:greedy} and the optimal objective value of~\ref{model:Knapsack}. The shaded area represents the 90\% confidence interval for the average approximation ratio. A larger approximation ratio indicates better performance; in particular, a ratio of 1 indicates that the algorithm delivers an optimal solution.

\begin{figure}[!h]
	\centering
 		\caption{Performance of Algorithm \ref{alg:greedy} in all instances of dataset~\texttt{A}.}
	\label{fig:greedy}
 \includegraphics[width=0.5\textwidth]{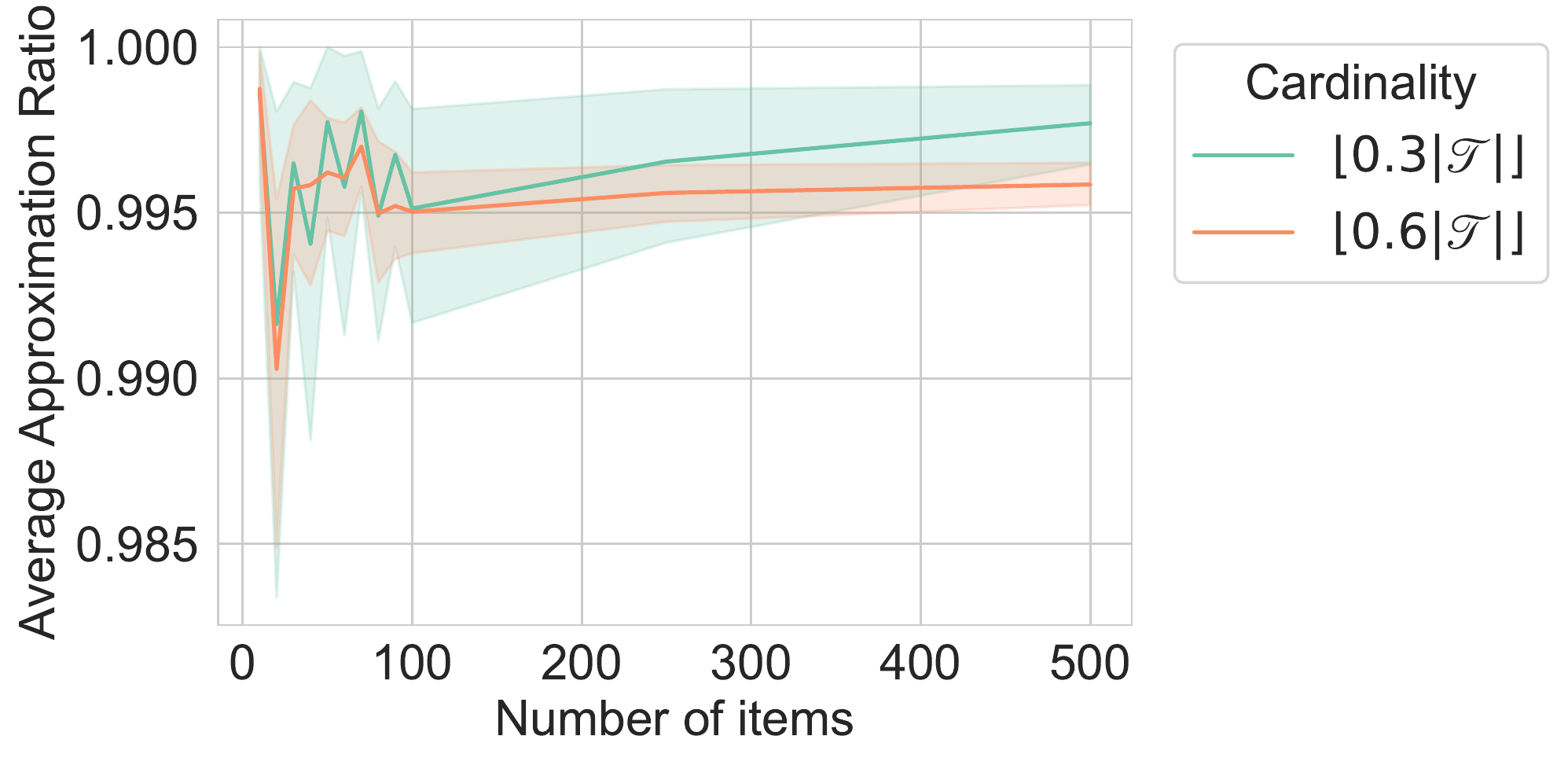}
\end{figure}

The results show strong performance of Algorithm \ref{alg:greedy}. The average approximation ratio is above 0.995, and the minimum ratio across all instances is larger than 0.95. We omit the results for~$\cardinalitycapacity = 2$ because Algorithm~\ref{alg:greedy} delivers optimal solutions for all instances of dataset~\texttt{A} in these cases.

\subsection{Empirical Performance of the Online Algorithm}
To evaluate Algorithm \ref{alg:online}, we construct dataset~\texttt{B} to test the performance of~$\online$ in scenarios where the optimal solutions will always consist of a single item. In dataset~\texttt{B}, each item~$\slot$ consists of two components, and we use the same procedure in~\texttt{A} to generate the first~$|\slots|-1$ items. The knapsack capacity~$\knapsackcapacity$ is also defined as in~\texttt{A}, but only based on the first~$|\slots| - 1$ items. For the last item~$\slot'$, the first component has weight~$a_{\slot',1}$ sampled uniformly from~$\left[5,0.49W\right]$ and utility~$\utility_{\slot',1} = 10a_{\slot',1}$, and the second component has weight~$a_{\slot',2} = \knapsackcapacity  -  a_{\slot',1}$ and utility~$\utility_{\slot',2} =  7a_{\slot',1}$. Note that the maximum utility density of the first~$|\slots|-1$ items is at most 5, so by construction, \textit{the optimal solution for any instance of \texttt{B} consists solely of~$\slot'$}, which occupies the entire knapsack capacity. This contrasts with dataset~\texttt{A}, where the utilities of the items are relatively uniform, and optimal solutions are more likely to have several items. The cardinality capacities are also extracted from~$\cardinalitycapacity \in \{2,\lfloor 0.3 |\slots| \rfloor,\lfloor 0.6 |\slots| \rfloor\}$, so dataset~\texttt{B} has 360 instances. 

For each instance of datasets \texttt{A} and \texttt{B}, we sample uniformly at random 20 arrival orders, so we report the results from 14,400 executions of~$\online$. Moreover, instead of using an approximation algorithm to solve the offline instance~$\Problem{}(\slots_l)$ (in line 8 of algorithm~\ref{alg:online}), we use an exact mixed-integer linear programming formulation and solve the problem to optimality using Gurobi. Finally, we note that the worse-case behavior of Algorithm \ref{alg:online} can be improved if we force the algorithm to incorporate as much from the last item as the remaining capacity allows. However, this modification does not allow us to improve the performance guarantees of our algorithm, so 
we evaluate the performance of~$\online$ without incorporating this modification in our experiments.

We present in Figure~\ref{fig:online} the results of our experiments involving Algorithm \ref{alg:online}.  Each plot in Figure~\ref{fig:online} reports the \textit{average empirical ratio} $\frac{\mathcal{A}(I)}{O^*(I)}$ together with the 90\% confidence interval ($y$-axis), where $O^*(I)$ is the optimal offline solution value, and $\mathcal{A}(I)$ is the objective value of the solution given by Algorithm \ref{alg:online} for instance $I$. We present the average empirical ratio instead of the competitive ratio (presented in Definition~\ref{def:ratio}) because Algorithm \ref{alg:online} may deliver solutions with zero utility (e.g., when it does not pick any item), for which the competitive ratio is undefined. Therefore, a larger empirical ratio indicates better performance, with 0 and 1 being the minimum and the maximum achievable values, respectively. Finally, the results are aggregated by~$|\slots|$ ($x$-axis), and each curve is associated with a cardinality capacity.

\begin{figure}[!h]
	\centering
 		\caption{Performance of Algorithm \ref{alg:online} in all instances of datasets~\texttt{A} and~\texttt{B}.}
	\label{fig:online}
	\begin{subfigure}[b]{0.45\textwidth}
	\includegraphics[width=1.0\textwidth]{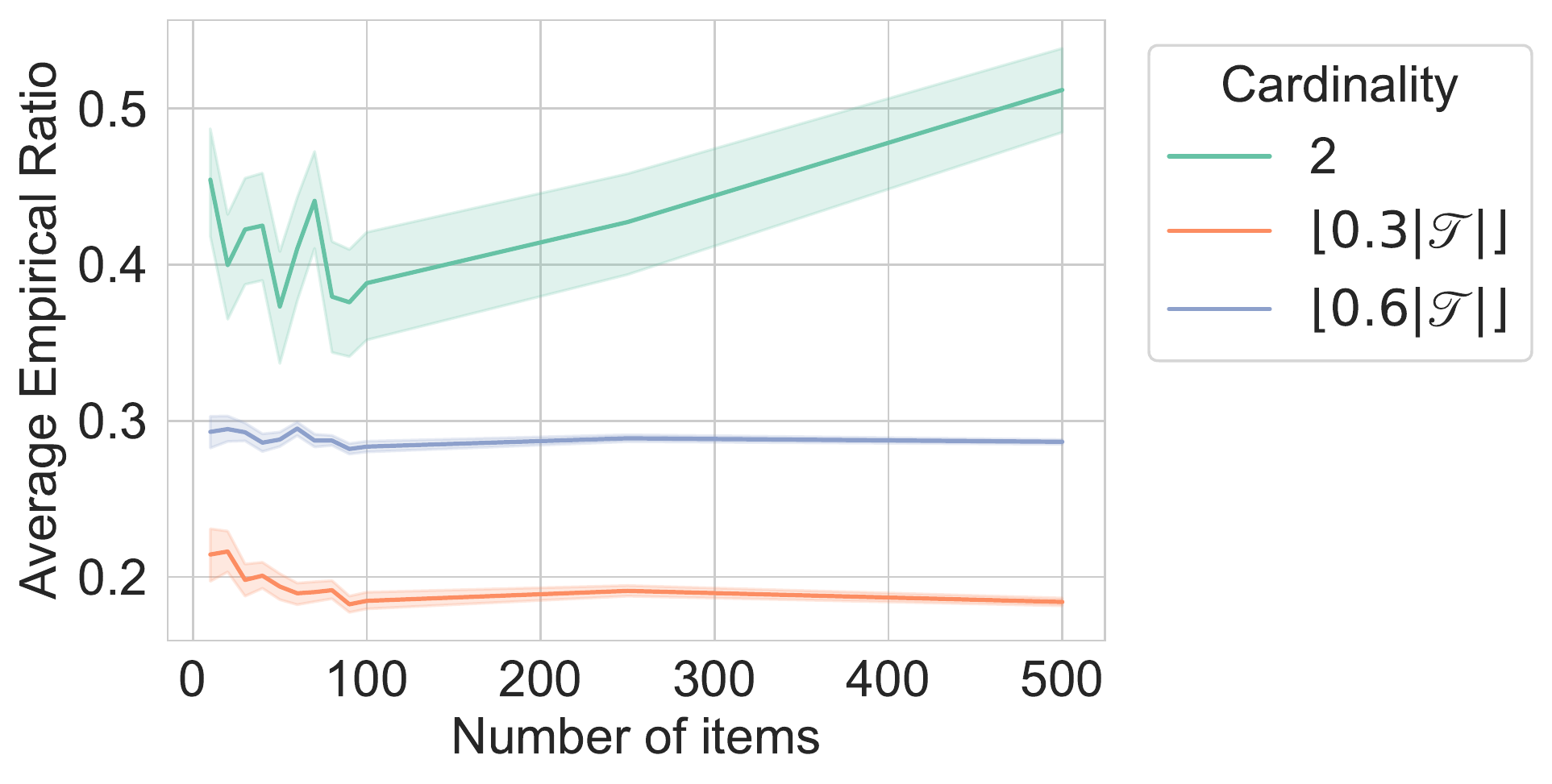}
	\subcaption{Dataset \texttt{A}}
	\end{subfigure}
		\begin{subfigure}[b]{0.45\textwidth}
     	\includegraphics[width=1.0\textwidth]{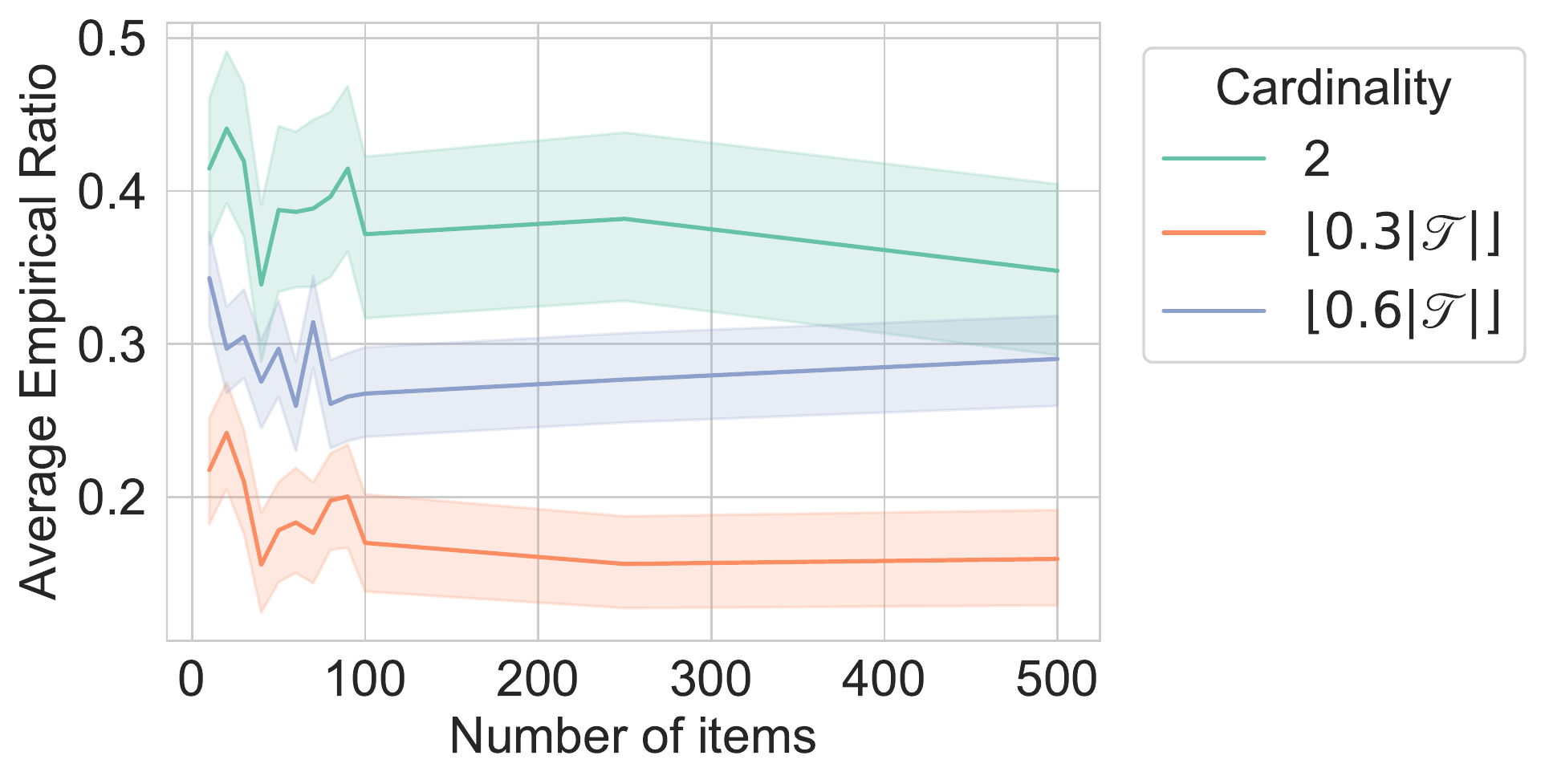}
	\subcaption{Dataset \texttt{B}}
	\end{subfigure}%
\end{figure}

Overall, the results show that the empirical performance of Algorithm \ref{alg:online} when~$\cardinalitycapacity = 2$ is clearly better than the theoretical ratio of 0.156 presented in Corollary~\ref{cor:ratioC2}, especially for dataset \texttt{A}. As expected, the results for dataset \textit{B} are slightly worse, and this happens because the value of~$c$ tailored for~$\cardinalitycapacity = 2$ makes~$\online$ ignore the most valuable item arrives during the sampling phase (and is thus ignored by the algorithm) in~37\% of the cases. Moreover, the knapsack phase does not trigger the incorporation of any item in more than 90\% of the cases (i.e., once the most valuable item appears, the knapsack phase always returns the optimal offline solution). Therefore, a ``hit-or-miss'' aspect is associated with the instances of dataset \texttt{B}, as shown in the histogram presented in Figure~\ref{fig:histograms}. Still, this worst-case behavior is compensated by excellent performance in other instances.

\begin{figure}[!h]
	\centering
 		\caption{The histogram of empirical ratios of Algorithm \ref{alg:online} for datasets~\texttt{A} and~\texttt{B}.}
	\label{fig:histograms}
	\begin{subfigure}[b]{0.45\textwidth}
	\includegraphics[width=1.0\textwidth]{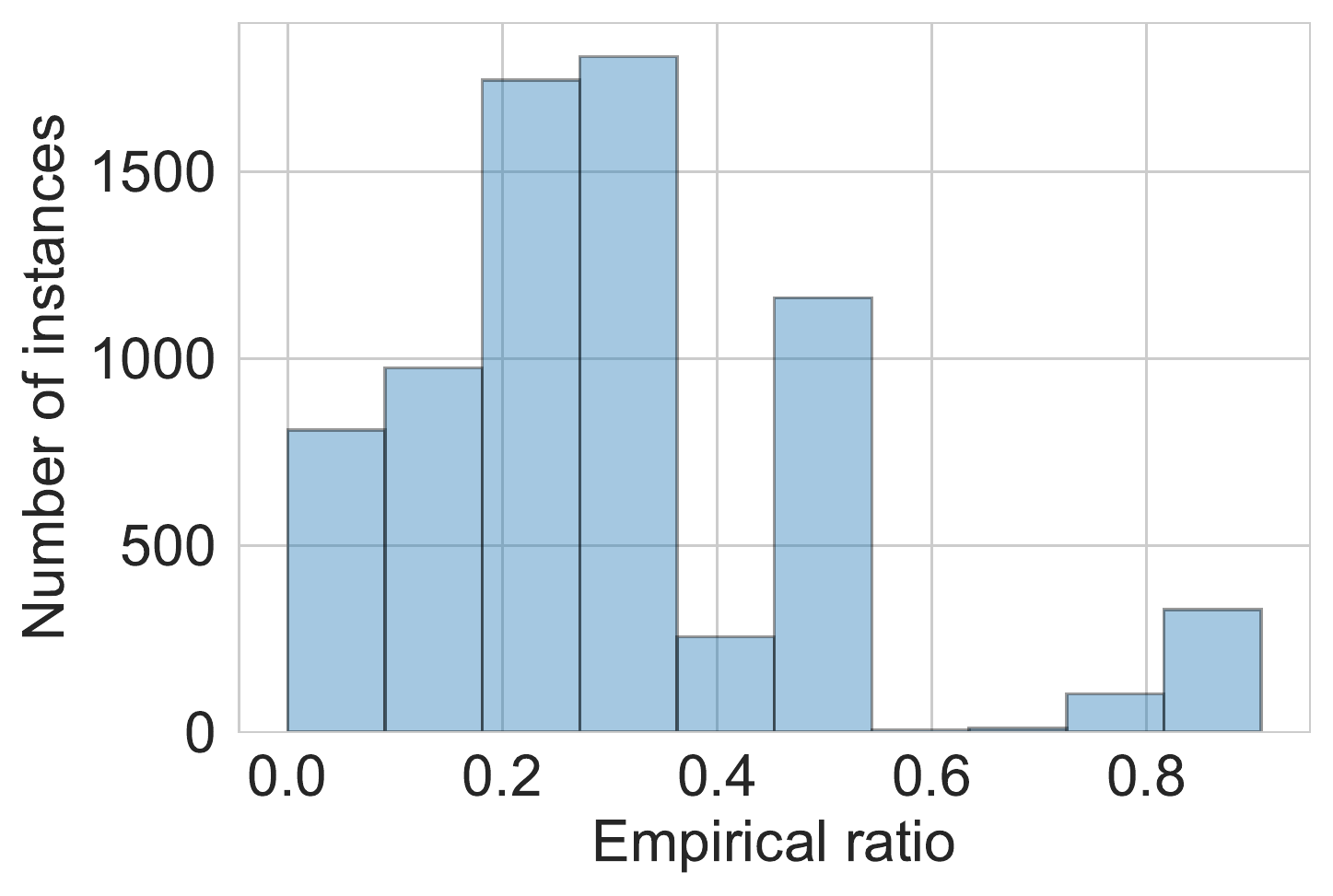}
	\subcaption{Dataset \texttt{A}}
	\end{subfigure}
		\begin{subfigure}[b]{0.45\textwidth}
     	\includegraphics[width=1.0\textwidth]{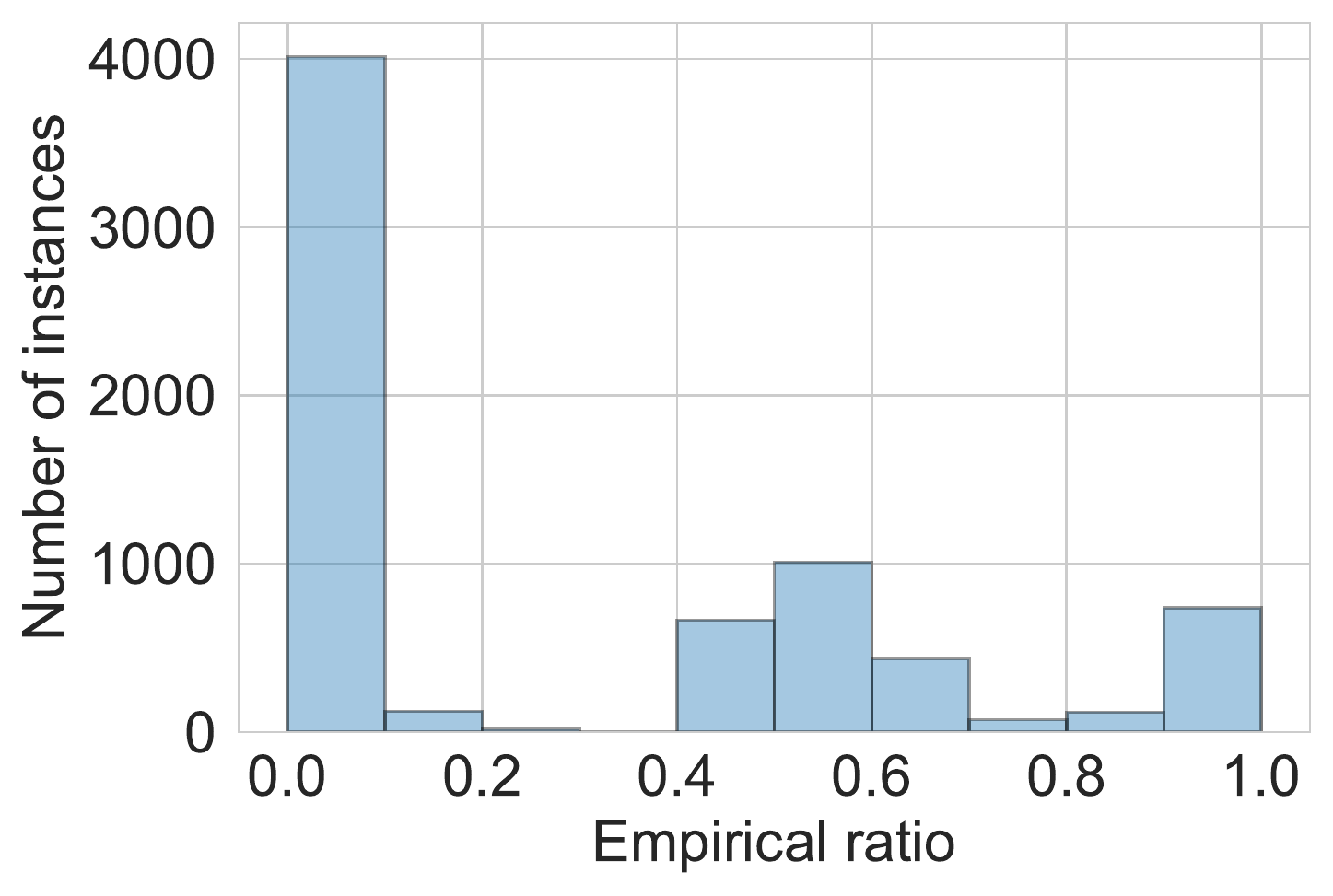}
	\subcaption{Dataset \texttt{B}}
	\end{subfigure}%
\end{figure}

The empirical performance of~$\online$ for~$\cardinalitycapacity = \lfloor 0.3 |\slots| \rfloor$ and~$\cardinalitycapacity = \lfloor 0.6 |\slots| \rfloor$ (presented in the curves labeled~30\% and~60\% in Figure~\ref{fig:online}, respectively) is approximately two times better than the theoretical ratios of 0.096 and 0.156 derived from our analysis, respectively. Moreover, the performance of~$\online$ is relatively uniform across both datasets in these cases (differently from~$\cardinalitycapacity = 2$, where dataset \texttt{A} is easier than dataset \texttt{B}). This can be explained by the fact that the higher cardinality capacity allows~$\online$ to incorporate more items, so several smaller ones may be used to partially recover the utility that is lost when high-value items arrive late.

\section{Conclusions}\label{sec:conclusions}

This article introduces and investigates the cardinality-constrained continuous knapsack problem with concave piecewise linear utilities. We develop an FPTAS and show that the~\Problem{} has a submodular structure that ensures a~$(1 - 1/e)$-approximation guarantee for a greedy procedure. We also present a $\frac{10.427}{\alpha}$-competitive algorithm for an online extension of the~\Problem{} in the random order model, where~$\alpha$ is the approximation guarantee of any algorithm used to solve offline instances of the~\Problem{} as part of our algorithm; additionally, we tailor our analysis to obtain stronger competitive ratios in special cases of the problem.
Our work contributes to the literature on offline knapsack problems by incorporating constraints and utility structures motivated by business and economic challenges. Additionally, our work provides a rigorous analysis of an algorithm for the online extension of the problem in the random order model. Results of online knapsack problems in the random order model are very recent. Our article is among the first to incorporate multiple constraints and nonlinear utilities and contributes to this growing research area.

\bibliographystyle{apalike}
\bibliography{references}

\newpage

\begin{APPENDICES}

\section{Proofs of Results in Section~\ref{sec:piecewise}}
\label{sec:reformulation}

Let the continuous variable~$x_{\slot,i}$ denote the utilization of component~$\component_{\slot,i}$. The utilization of item $\slot$ can be rewritten as $x_{\slot}=\sum\limits_{i = 1}^{\numcomponents_{\slot}}x_{\slot,i}$. Finally, we  linearize~$\| x_{\slot} \|_0$ by introducing binary decision variables $y_{\slot}$  to indicate the activation of $x_{\slot}$, that is, $y_{\slot}=1$ if~$x_{\slot} > 0$ and $y_{\slot}=0$ otherwise. These transformations allow us to reformulate~\ref{model:original} as the following MIP:  
\[
\begin{array}{ccll}\label{model:ConstAppSch}
\tag{\textbf{\texttt{CKP-r}}}
 &\max
& \sum\limits_{\slot \in \slots} \sum\limits_{i=1}^{\numcomponents_{\slot}} \utility_{\slot,i} \cdot x_{\slot,i}  
\\
&(a) & \sum\limits_{\slot \in \slots} \sum\limits_{i=1}^{\numcomponents_{\slot}}  x_{\slot,i}   \leq \knapsackcapacity \\
&(b)& \sum\limits_{\slot \in \slots} y_{\slot} \leq \cardinalitycapacity  & \\
&(c)&x_{\slot,i} \leq \maxcapacity_{\slot,i} \cdot y_{\slot} &\forall (\slot,i) \in \slots \times \{1,2,\ldots,\numcomponents_{\slot}\} \\
&& y_{\slot}\in \{0,1\} & \forall \slot \in \slots \\
&& x_{\slot,i} \in \mathbb{R}^+ & \forall (\slot,i) \in \slots \times \{1,2,\ldots,\numcomponents_{\slot}\}
\end{array}
\]
Our decomposition strategy represents the utility of each item~$\slot$ as the linear expression $\sum\limits_{i = 1}^{\numcomponents_{\slot}}\utility_{\slot,i} x_{\slot,i}$. Constraints~\ref{model:ConstAppSch}-(a) and~\ref{model:ConstAppSch}-(b) enforce the knapsack and  cardinality constraints, respectively. Constraint~\ref{model:ConstAppSch}-(c) connects the utilization of each component~$\component_{\slot,i}$ with the respective activation variable~$y_{\slot}$ and capacity~$\maxcapacity_{\slot,i}$; particularly, $y_{\slot}$ must be equal to 1 if any component of~$\slot$ has non-zero utilization, so it models~$\| x_{\slot} \|_0$. We abuse notation and use~$x_{\slot,i}$ and~$x_{\slot}$ when referring to the utilization of component~$\component_{\slot,i}$ and item~$\slot$, respectively; note that $x_{\slot} = \sum\limits_{i = 1}^{\numcomponents_{\slot}} x_{\slot,i}$.  Proposition \ref{prop:eqiv} shows the equivalence of Formulations \ref{model:original} and \ref{model:ConstAppSch}.

We first show the following two lemmas in order to prove Proposition \ref{prop:eqiv}.
\begin{lemma}\label{lemma:seq1}
    Optimal solutions to \ref{model:ConstAppSch} must satisfy the following properties:
\begin{align}
&x_{\slot,i}>0 \Longrightarrow x_{\slot,i-1}=\maxcapacity_{\slot,i-1}, &\forall i>1, \slot\in \slots \label{eq:seq}
\end{align}
\end{lemma}

\begin{proof}{Proof of Lemma \ref{lemma:seq1}:}
This result can be shown by contradiction. Suppose that an optimal solution $\vecx^0$ to~\ref{model:ConstAppSch} violates condition Eq.\eqref{eq:seq}; that is, we have some $x^0_{\slot,i}>0$ and $x^0_{\slot,i-1}<\maxcapacity_{\slot,i-1}$. We can construct a feasible solution $\vecx^1$ with $(x^1_{\slot,i-1}, x^1_{\slot,i})=(x^0_{\slot,i-1}+\delta, x^0_{\slot,i}-\delta)$ where $\delta=\min(\maxcapacity_{\slot,i-1}-x^0_{\slot,i-1},x^0_{\slot,i})>0$. As $\utility_{\slot,i-1} > \utility_{\slot,i}$, $\vecx^1$ gives a larger objective value than $\vecx^0$, which contradicts the optimality of $\vecx^0$. This concludes the proof of Lemma~\ref{lemma:seq1}.
\medskip
\end{proof}

\begin{lemma}\label{lemma:seq2} For every instance of the~\Problem{}, there exists an optimal solution to~\ref{model:ConstAppSch} such that
\begin{align}
&y_{\slot}=0 \Longleftrightarrow \sum\limits_{i= 1}^{\numcomponents_{\slot}} x_{\slot,i}=0, &\forall \slot\in \slots \label{eq:activate}
\end{align}
\end{lemma}
\begin{proof}{Proof of Lemma \ref{lemma:seq2}:}
Let~$(\vecx,\vecy)$ be an optimal solution that does not satisfy the condition Eq.\eqref{eq:activate}. First, from Constraint~\ref{model:ConstAppSch}-(c), we have that $y_{\slot}=0$ implies $\sum\limits_{i= 1}^{\numcomponents_{\slot}} x_{\slot,i}=0$. Moreover, if $\sum\limits_{i= 1}^{\numcomponents_{\slot}} x_{\slot,i}=0$ and~$y_{\slot} = 1$, we can obtain a solution that achieves the same objective value by setting~$y_{\slot} = 0$ without changing~$\vecx$. 
This concludes the proof of Lemma~\ref{lemma:seq2}.
\medskip
\end{proof}

Equipped with the lemmas above, we can prove the equivalence of Formulations \ref{model:original} and \ref{model:ConstAppSch}.
\begin{proposition}
\label{prop:eqiv}
    \ref{model:original} and \ref{model:ConstAppSch} are equivalent, in the sense that their optimal objective values are the same and optimal solutions to one can be converted into optimal solutions to the other.
\end{proposition}

\begin{proof}{Proof of Proposition~\ref{prop:eqiv}:}
Any optimal solution to \ref{model:ConstAppSch} that satisfies the conditions in Lemma~\ref{lemma:seq2} can be converted into a feasible solution to \ref{model:original} by constructing $x_{\slot}=\sum\limits_{i= 1}^{\numcomponents_{\slot}} x_{\slot,i}$, which satisfies the structure of $\utilities_{\slot}(\cdot)$ because of condition Eq.\eqref{eq:seq}. Moreover, due to condition Eq.\eqref{eq:activate}, we have $y_{\slot}=0$ if $\sum\limits_{i= 1}^{\numcomponents_{\slot}} x_{\slot,i}=0$ and  $y_{\slot}=1$ if $\sum\limits_{i= 1}^{\numcomponents_{\slot}} x_{\slot,i}>0$. Thus, $y_{\slot}$ in these optimal solutions to \ref{model:ConstAppSch} is equivalent to $\| x_{\slot} \|_0$. Finally, $(\vecx,\vecy)$ satisfy the constraints and attain the same objective value in~\ref{model:original}, so the optimal objective value of~\ref{model:ConstAppSch} is less than or equal to that of~\ref{model:original}.

Next, we show that any feasible solution $\vecx$ to \ref{model:original} can be converted into a feasible solution to \ref{model:ConstAppSch}. We construct $x_{\slot,1}=\min(\maxcapacity_{\slot,i},x_{\slot})$ and $x_{\slot,i}=\min(\maxcapacity_{\slot,i},x_{\slot}-\sum\limits_{i' = 1}^{i-1} x_{\slot,i'})$, $\forall 1<i<\numcomponents_{\slot}$. Additionally, we set $y_{\slot}=1$ if $x_{\slot}>0$ and $y_{\slot}=0$ if $x_{\slot}=0$. Direct substitution shows that the constructed solution is feasible and attains the same objective value of~$\vecx$ to \ref{model:ConstAppSch}. Therefore, the optimal objective value of~\ref{model:original} is less than or equal to that of~\ref{model:ConstAppSch}. It follows that the optimal objective values of~\ref{model:original} and~\ref{model:ConstAppSch} are the same. 

Based on the conversion discussed above, we can convert any optimal solutions to~\ref{model:original} to optimal solutions to \ref{model:ConstAppSch} satisfying conditions in Lemma~\ref{lemma:seq2}. If the cardinality constraint in \ref{model:ConstAppSch} is not binding at these optimal solutions, i.e., $\sum\limits_{\slot \in \slots} y_{\slot} < \cardinalitycapacity$, we can construct alternative optimal solutions to \ref{model:ConstAppSch} by allocating the extra cardinality capacity by making some $y_{\slot}=0$ to $y_{\slot}=1$, which does not change the objective value. For any optimal solutions to \ref{model:ConstAppSch}, we can convert them to optimal solutions to~\ref{model:original} by constructing $x_{\slot}=\sum\limits_{i= 1}^{\numcomponents_{\slot}} x_{\slot,i}$ that attains the same objective value. Therefore, optimal solutions to one model can be converted to optimal solutions to the other model. 
\medskip
\end{proof}

\begin{proof}{Proof of Remark~\ref{prop:partialcomponent}:} We use the notation of Formulation~\ref{model:ConstAppSch} to 
prove Remark~\ref{prop:partialcomponent}. Note that in any optimal solution~$\vecx^*$, a component~$\component_{\slot,i}$ is utilized if~$x_{\slot',i'} > 0$ and partially utilized if $0 < x_{\slot,i} < \maxcapacity_{\slot,i}$, and a utilized. 

We prove the first result by contradiction. Suppose that for any optimal solution $(\vecx^*,\vecy^*)$ to an instance of the~$\Problem$, we have at least two components being partially utilized. Let~$\component_{\slot,i}$ and~$\component_{\slot',i'}$ be two components being partially utilized, i.e., $0 < x^*_{\slot,i}< \maxcapacity_{\slot,i}$ and~$0 < x^*_{\slot',i'} < \maxcapacity_{\slot',i'}$. We assume w.l.o.g. that $\utility_{\slot,i} \geq \utility_{\slot',i'}$. We construct an alternative solution~$\vecx'$ such that~$x'_{\slot,i} = x^*_{\slot,i} + \underline{\epsilon}$ and $x'_{\slot',i'} = x^*_{\slot',i'} - \underline{\epsilon}$,  where~$\underline{ \epsilon} = \min(\maxcapacity_{\slot,i}-x^*_{\slot,i}, x^*_{\slot',i'})$; all other component utilization in~$\vecx'$ are the same as in~$\vecx^*$, and the values of $y_{\slot}^*$ are unchanged.  The objective value of~$\vecx'$ increases the objective of~$\vecx^*$  by~$\underline{\epsilon}(\utility_{\slot,i} - \utility_{\slot',i'})$; from the optimality of~$\vecx^*$, we must have $\utility_{\slot,i} = \utility_{\slot',i'}$. As a result, $\vecx'$ gives the same objective value as $\vecx^*$ but has at most one component partially utilized due to the change $\underline{ \epsilon}$, which leads to a contradiction. The proof of the second result uses similar arguments, i.e., a contradiction can be proven through the same substitution process.
\end{proof}

\section{Proofs of Results in Section \ref{sec:FPTAS1}}

\begin{proof}{Proof of Proposition~\ref{prop:greedyoptimal}:}  We show that the identification of an optimal solution to an instance~$I$ of the~$\Problem{}^T$ 
is equivalent to solving an instance~$I'$ of the continuous knapsack problem where the knapsack capacity~$\knapsackcapacity$ is the same as in~$I$ and each component~$\component_{\slot,i}$ of~$I$ becomes an item in~$I'$ with capacity~$\maxcapacity_{\slot,i}$ and utility~$\utility_{\slot,i}$. The following greedy algorithm can identify an optimal solution to~$I'$. We construct a 
sequence~$\components \coloneqq (\component'_{1}$,  $\component'_{2}, \ldots, )$ of items of~$I'$, sorted by utility in descending order. Given~$\components$, we solve the problem iteratively by picking in each step the unused item with the highest utility, using as much of its capacity as possible; observe that the consumption of an item is bounded either by its capacity~$\maxcapacity_{\slot,i}$ or by the remaining knapsack capacity, which is updated in each step. 

As~$\utility_{\slot,i} > \utility_{\slot,i+1}$ for $i \in \{1,\ldots,\numcomponents_{\slot}-1\}$, the components of the same item are selected sequentially by the greedy algorithm. Finally, standard arguments show that any solution that does not follow the structure of the solutions produced by the greedy algorithm is sub-optimal, so it follows that the~$\Problem^T$ can be solved efficiently.    
\end{proof}

\section{Proofs of Results in Section \ref{sec:submodularity}}

\begin{proof}{Proof of Proposition \ref{prop:pointwisesubmodular}:} 
For a set of items~$\slots$, let~$\numcomponents(\slots) = \sum\limits_{\slot \in \slots}\numcomponents_{\slot}$ be the total number of components associated with items in~$\slots$, and let~$\components(\slots) \coloneqq \Big(\component_{\slots}(1)$,  $\component_{\slots}(2), \ldots, \component_{\slots}(\numcomponents(\slots))\Big)$ denote the sequence of these components sorted in descending order of utility. Also, we use~$\utility_{\slots}(i)$  and~$\maxcapacity_{\slots}(i)$ to denote the  utility and the capacity of~$\component_{\slots}(i)$, respectively.

According to Remark~\ref{prop:partialcomponent} and Proposition~\ref{prop:greedyoptimal}, there exists an optimal solution~$\vecx$ to~$G(\slots)$ that utilizes the components of items in $\slots$ in descending order of their utility, with at most one of these components being only partially utilized. For convenience, we define a vector~$\vecz(\vecx) \in \mathbb{R}^{\numcomponents(\slots)}$ indexed by the elements of~$\components(\slots)$ to represent the utilization of each component~$\component_{\slots}(i) \in \components(\slots)$, i.e., $\vecz$ is a permutation of~$\vecx$ based on component utility. More precisely, given~$\vecx$, we define~$\vecz(\vecx)$ as 
{
\footnotesize
\begin{equation*}
    \begin{cases}
    z_i(\vecx)=\maxcapacity_{\slots}(i) & \text{if }   i<l(\vecx),\\
    z_i(\vecx)\in (0,\maxcapacity_{\slots}(i)] & \text{if } i=l(\vecx),\\
    z_i(\vecx)=0 & \text{if }   i>l(\vecx),
    \end{cases}
\end{equation*}}
\noindent where $l(\vecx) \equiv \max\{i:z_i(\vecx) > 0, i =1,2,\ldots,n(\slots)\}$ is the number of utilized components in~$\vecx$. 

We use the ordering of~$\components(\slots)$ to define a non-increasing step function~$V_{\slots,\vecx}(t)$, which gives the utility of the component to be used once~$t$ capacity has been distributed to the first elements of~$\components(\slots)$. More precisely, $V_{\slots,\vecx}(t)$ is computed as follows 
{
\footnotesize
\begin{equation*}
     V_{\slots,\vecx}(t) =\begin{cases}
    \utility_{\slots}(1) & \text{if } 0<t\leq z_{1}(\vecx),\\
    \utility_{\slots}(2) & \text{if } z_{1}(\vecx)<t\leq z_{1}(\vecx)+z_{2}(\vecx),\\
    \dots \\
    \utility_{\slots}(k) & \text{if } \sum\limits_{i = 1}^{k-1}z_{i}(\vecx)<t\leq \sum\limits_{i = 1}^{k}z_{i}(\vecx),\\
     \dots \\
      \utility_{\slots}(l(\vecx)) & \text{if } \sum\limits_{i = 1}^{l(\vecx)-1}z_{i}(\vecx)<t\leq \sum\limits_{i = 1}^{l(\vecx)}z_{i}(\vecx),\\
    0 & \text{if } t > \sum\limits_{i = 1}^{l(\vecx)}z_{i}(\vecx).
    \end{cases}
\end{equation*}
}
Observe that $V_{\slots,\vecx}(t)$ consisting of $l(\vecx)+1$ steps, and we define $V_{\slots,\vecx}(t)=0$ for every $t \geq \| \vecx \|_1$.

Next, we proceed with the proof. We consider an arbitrary instance of the $\Problem{}^T$  parameterized by a set of items~$\slots$. Let~$\slot$ be an item of~$\slots$, and let $\underline{ \slots}$ and~$\overline{\slots}$ be subsets of~$\slots \setminus \{\slot\}$ such that~$\underline{\slots} \subset \overline{\slots}$. 

\noindent\textbf{\underline{Monotonicity}} 
The non-decreasing monotonicity of~$G(\cdot)$ follows from the fact that any feasible solution to the~\Problem{} parameterized by~$\underline{\slots}$ is also feasible and has the same objective value in the~$\Problem{}^T$ parameterized by~$\overline{\slots}$. 

\noindent\textbf{\underline{Submodularity}} 
    Let~$\underline{\vecx}$, $\overline{\vecx}$, 
    $\underline{\vecx'}$, and~$\overline{\vecx'}$
    denote optimal solutions to~$\underline{\slots}$, $\overline{\slots}$, $\underline{\slots} \cup \{\slot\}$, and~$\overline{\slots} \cup \{\slot\}$,  respectively. According to Remark~\ref{prop:partialcomponent}, $\underline{\vecx}$ and~$\underline{\vecx'}$ differ if~$\slot$ has one or more components whose utilities are larger than the utility of some components with non-zero utilization in~$\underline{\vecx}$. Namely, this happens if the utility~$ \utilities_{\slot} (\underline{x'_{\slot}})$ gained by the utilization of item $\slot$ in~$\underline{\vecx'}$ is larger than~$\int\limits_{\knapsackcapacity-\underline{x'_{\slot}}}^{\knapsackcapacity}V_{\underline{\slots},\underline{\vecx}}(t)dt$.   Otherwise, if the incorporation of~$\slot$ does not lead to improvement, we have $\underline{x_{\slot}'} = 0$. Therefore, the utility gains~$\Delta(\underline{\slots},\slot)$ brought by the incorporation of~$\slot$ to~$\underline{\slots}$  is given by 
    \[
    \Delta(\underline{\slots},\slot) 
    \equiv 
    G(\underline{\slots} \cup \{\slot\}) - 
    G(\underline{\slots}) 
    = 
  \utilities_{\slot} (\underline{x'_{\slot}})
    - \int\limits_{\knapsackcapacity-\underline{x_{\slot}'}}^{\knapsackcapacity}V_{\underline{\slots},\underline{\vecx}}(t) dt.
    \]
    Similarly, $\Delta(\overline{\slots},\slot)$  is given by
    \[
    \Delta(\overline{\slots},\slot) \equiv 
    G(\overline{\slots} \cup \{\slot\}) - 
    G(\overline{\slots}) = 
  \utilities_{\slot} (\overline{x'_{\slot}})
    -
    \int\limits_{\knapsackcapacity-\overline{x_{\slot}'}}^{\knapsackcapacity}V_{\overline{\slots},\overline{\vecx}}(t) dt.
    \]
    Therefore, we have 
    \begin{eqnarray*}
    \Delta(\overline{\slots},\slot) 
    &=&
        \utilities_{\slot} (\overline{x'_{\slot}})
        - 
        \int\limits_{\knapsackcapacity-\overline{x'_{\slot}} }^{\knapsackcapacity}V_{\overline{\slots},\overline{\vecx}}(t) dt \\
    &\leq&
        \utilities_{\slot} (\overline{x'_{\slot}})
        - 
        \int\limits_{\knapsackcapacity-\overline{x'_{\slot}} }^{\knapsackcapacity}V_{\underline{ \slots},\underline{\vecx}}(t) dt \\
    &\leq&    
       \utilities_{\slot} (\underline{x'_{\slot}})
        - 
        \int\limits_{\knapsackcapacity-\underline{x_{\slot}'}}^{\knapsackcapacity}V_{\underline{\slots},\underline{\vecx}}(t) dt
    =
        \Delta(\underline{\slots},\slot). 
    \end{eqnarray*}
    The first inequality follows from the fact that~$V_{\overline{\slots},\overline{\vecx}}(t) \geq V_{\underline{\slots},\underline{\vecx}}(t)$ for every~$t$ in~$[0,\knapsackcapacity]$. Next, we show that the second inequality holds. As~$\components(\underline{\slots} \cup \{\slot\}) \subset \components(\overline{\slots} \cup \{\slot\})$, the position of component~$\component_{\slot,k}$ in~$\components(\underline{\slots} \cup \{\slot\})$     is not greater than its position in~$\components(\overline{\slots} \cup \{\slot\})$. Therefore, it follows from Proposition~\ref{prop:greedyoptimal} that $\underline{x_{\slot}'} \geq \overline{x_{\slot}'}$ and the second inequality must hold; otherwise, one could construct a  solution to~$\underline{\slots} \cup \{\slot\}$ that is strictly better than~$\underline{\vecx'}$ by reducing the utilization of $\slot$ from~$\underline{x'_{\slot}}$ to~$\overline{x'_{\slot}}$,   thus contradicting the optimality of~$\underline{\vecx'}$.   Therefore, we have $\Delta(\overline{\slots},\slot) \leq \Delta(\underline{\slots},\slot) $ for~$\underline{\slots} \subset \overline{\slots} \subseteq \slots \setminus \{\slot\}$ and thus $G(\cdot)$ is submodular. 
\end{proof}

\section{Continuous Utility Functions}\label{sec:Continuous Utilities}

\begin{theorem}\label{thm:FPTAS for smooth generalization}
    The generalization of the~$\Problem$ where the utility functions are injective, continuous, and Lipschitz continuous admits an FPTAS.
\end{theorem}
\begin{proof}{Proof of Theorem~\ref{thm:FPTAS for smooth generalization}:} The dynamic programming formulation in this case  computes~$d(i,l,\fullutility)$ for triples in $\mathcal{M} \equiv \{0, 1,\ldots, \numitems\} \times \{0, 1,\ldots,\cardinalitycapacity\} \times \left[0,\sum\limits_{\slot' \in \slots} \utilities_{\slot'}(\maxcapacity_{\slot}')  \right]$. The state space~$\mathcal{M}$ is similar to the one defined in~\S\ref{sec:FPTAS1}, except that the first coordinate is defined over items (rather than components). The interpretation of the states and the initialization procedures for~$d$ and~$\fullutility^*$ are as in Algorithm~\ref{algo:dp}. Namely, $d(i,l,\fullutility)$ is the smallest amount of knapsack capacity used by a feasible solution to~$\Problem{}$ containing~$l$ out of the first $i$ items in~$\slots$ attaining objective value~$\fullutility$, based on any arbitrary ordering of~$\slots$.  Moreover, we use~$d(i,l,\fullutility) = \knapsackcapacity+1$ if there is no set of~$l \leq i$ items attaining~$\fullutility$. The construction of~$d$ relies on an iterative procedure to sequentially incorporate items in~$\slots$. The value of~$d(i,l,\fullutility)$ is computed by the following recursive expression:
\begin{equation}\label{eq:dpexpression}
d(i,l,\fullutility) = \min\limits_{\fullutility' \in [0,\mu] }\left\{d(i-1,l,\fullutility), d(i-1,l-1,\fullutility - \fullutility') + \utilities^{-1}_{\slot_i}(\fullutility') \right\}.    
\end{equation}
We use~$\utilities^{-1}_{\slot_i}(\fullutility')$ to denote the inverse function of~$\utilities_{\slot_i}$, i.e.,  $\maxcapacity_{\slot_i}' \coloneqq \utilities^{-1}_{\slot_i}(\fullutility')$ is the utilization of~$\slot_i$ in $[0,\maxcapacity_{\slot_i}]$ such that~$\utilities_{\slot_i}(\maxcapacity_{\slot_i}') = \fullutility'$. The main challenge is the estimation of~$\utilities^{-1}_{\slot_i}(\fullutility')$. Namely, if~$\utilities^{-1}_{\slot_i}(\fullutility')$ can be exactly computed, the discretization procedure and the approximation guarantees are similar to the ones presented in~\S\ref{sec:FPTAS1}. Otherwise, we can explore the concavity of~$\utilities_{\slot}$ and estimate~$\maxcapacity_{\slot_i}' = \utilities^{-1}_{\slot_i}(\fullutility')$ for any given~$\fullutility'$ through binary search in~$[0,\maxcapacity_{\slot_i}]$.  The same challenges involving numeric precision apply to this procedure, though, so we may be forced to stop the search upon the identification of some interval~$[a,b]$ such that~$\maxcapacity_{\slot_i}' \in [a,b]$ and $b-a \leq \epsilon'$ for some~$\epsilon' > 0$. Once such an interval is identified, we set~$\maxcapacity_{\slot_i}' = a$ and finalize the search. If the utility functions 
are~$L$-Lipschitz continuous with respect to the~$L_1$ norm, we have
\[
\| \utilities_{\slot_i}(\maxcapacity_{\slot_i}) - \utilities_{\slot_i}(\maxcapacity_{\slot_i} - \epsilon') \|_1 \leq L(b-a) = \epsilon' L,
\]
and by adding up to~$\cardinalitycapacity$ numbers with such error, we have a maximum error of~$\epsilon =|\cardinalitycapacity| \epsilon' L$, i.e., the error in the objective function can be controlled through through~$\epsilon'$. Under such an assumption, we can replace the (infinite) set~$\Theta \coloneqq \left[0,\sum\limits_{\slot' \in \slots} \utilities_{\slot'}(\maxcapacity_{\slot}')  \right]$ defining the third dimension of~$\mathcal{M}$ for a discrete set~$\Theta_{\epsilon} = \left\{0,\delta,2\delta,\ldots, \lfloor \frac{n}{\epsilon} \rfloor\right\}$, where 
$\delta = \frac{\epsilon}{\numitems}\max\limits_{\slot \in \slots}\utilities_{\slot}(\maxcapacity_{\slot})$, and hence obtain an FPTAS as in~\S\ref{sec:FPTAS1}.
\medskip
\end{proof}

\section{Proofs of Results in Section~\ref{sec:knapsack}}

\begin{proof}{Proof of Lemma~\ref{lemma:cdbound}:} 
For any given set of items~$\slots_{d \cdot \numitems} \subseteq \slots$,
let~$\slots_{d \cdot \numitems}^*$ be the subset of items in~$\slots_{d \cdot \numitems}$ that attain the maximum total utility~$r_{d \cdot \numitems}^* \coloneqq \max\limits_{\slot \in \slots_{d \cdot \numitems}}\utilities(\maxcapacity_{\slot})$ across all items in~$\slots_{d \cdot \numitems}$, i.e., 
$\slots_{d \cdot \numitems}^* = \left\{\slot \in \slots_{d \cdot \numitems}: \utilities(\maxcapacity_{\slot}) = r_{d \cdot \numitems}^*\right\}$. If an item in~$\slots_{d \cdot \numitems}^*$ arrives in~$\slots_{c \cdot \numitems}$, the knapsack will be empty at the end of the secretary phase, as~$\online_{S}$ could only incorporate an item with utility strictly larger than~$r_{d \cdot \numitems}^*$. Otherwise, if there exists just one item~$\slot$ in~$\slots_{d \cdot \numitems}^*$  and~$\slot$ arrives in~$\slots_{c \cdot \numitems+1, d \cdot \numitems}$, then~$\online_{S}$ necessarily incorporates an item. 
From the uniformity of the arrivals, all permutations of~$\slots_{d \cdot \numitems}$ are equally likely to occur, so the occurrence probability of the event that a specific item~$\slot$ in~$\slots_{d \cdot \numitems}^*$ arrives in the first~$c \cdot \numitems$ steps is~$\frac{c}{d}$. Finally, as~$\slots_{d \cdot \numitems}$ is finite, $\slots_{d \cdot \numitems}^*$ must contain at least one element, and from our assumption that all items have distinct total utilities, it follows that~$|\slots_{d \cdot \numitems}^*| = 1$. Therefore, the knapsack is empty at the beginning of the knapsack phase with probability~$\frac{c}{d}$.
\end{proof}

Lemma~\ref{lemma:simple} presents basic properties of~$\overline{\Phi}_{l,\slot}$ used in our arguments.
\begin{lemma}\label{lemma:simple}
For every~$l$ in~$\{1,2,\ldots,\numitems\}$ and~$\slot$ in~$\slots$, $Pr(\overline{\Phi}_{l,\slot}=1)=\frac{1}{\numitems}$, and for a given~$\slots_{l}$,  $Pr(\overline{\Phi}_{l,\slot}=1|\slots_l)=\frac{1}{l}$ for every~$\slot$ in $\slots_l$. Moreover,
for~$l \geq dn+1$, those probabilities are independent of~$\xi$, i.e., 
we have $Pr(\overline{\Phi}_{l,\slot}=1|\xi)=\frac{1}{\numitems}$ and $Pr(\overline{\Phi}_{l,\slot}=1|\xi, \slots_l)=\frac{1}{l}$. 
\end{lemma}

\begin{proof}{Proof of Lemma \ref{lemma:simple}:}
 From the assumption on the uniformity of the arrival orders, item~$\slot\in \slots$ arrives at the~$l$-th position with probability~$\frac{1}{\numitems}$ and thus the first result holds. For any given~$\slots_{l}$, every item~$\slot \in \slots_l$ is equally likely to arrives at position~$l$, i.e., $Pr(\overline{\Phi}_{l,\slot}=1|\slots_l)=\frac{1}{l}, \forall \slot \in \slots_l$. 

 For~$l \geq d \cdot \numitems+1$, $Pr(\overline{\Phi}_{l,\slot}=1)$ and $Pr(\overline{\Phi}_{l,\slot}=1|\slots_l)$ are independent from the outcome of $\xi$. This observation holds because the outcome of $\xi$ is affected only by the relative arrival order of the items in~$\slots_{d \cdot \numitems}$, i.e., it does not depend on the items composing~$\slots_{d \cdot \numitems}$. (see Lemma~\ref{lemma:cdbound}). 
\end{proof}

\begin{proof}{Proof of Lemma \ref{lemma:lb_l}:}
Let~$\knapsackcapacity^{(l)}$ and~$\cardinalitycapacity^{(l)}$ be random variables indicating the amount of knapsack and cardinality capacities consumed by~$\online$ in step~$l$, respectively.  By definition, we have~$\knapsackcapacity_{l} =  \sum\limits_{k = 0}^{l-1}\knapsackcapacity^{(k)}$ and~$\cardinalitycapacity_{l} =   \sum\limits_{k = 0}^{l-1}\cardinalitycapacity^{(k)}$.

For a given~$\slots_l$ and the corresponding solution~$s^{(l)} = (\vecx^{(l)},\vecy^{(l)})$ to~$\Problem{}(\slots_l)$ obtained by $\Approx_{\alpha}(\slots_l)$, we must have~$\sum\limits_{\slot \in \slots_l}x^{(l)}_{\slot} \leq W$ and~$\sum\limits_{\slot \in \slots_l}y^{(l)}_{\slot} \leq \cardinalitycapacity$. Note that $s^{(l)}$ is the same for all arrival orders of $\slots_l$, i.e., $s^{(l)}$ is invariant to the permutation defining the arrival sequence of~$\slots_l$. 

 Furthermore, if $\online_{K}$ incorporates~$\slot_l$ in the knapsack phase ($l \geq dn+1$), the utilizations~$x_{\slot_l}$ and~$y_{\slot_l}$ of the knapsack and cardinality for item~$\slot_l$, respectively, are at most $\beta x^{(l)}_{\slot_l}$ and $y^{(l)}_{\slot_l}$, respectively (see line 10 of Algorithm~\ref{alg:online}).

Therefore, following Lemma \ref{lemma:simple}, among all arrival sequences that are permutations of a set~$\slots_l$, the expected utilizations of the knapsack and cardinality capacities at step $l \geq d \cdot \numitems + 1$ in the knapsack phase given $\xi = 1$ satisfy $\Ex[\knapsackcapacity^{(l)}|\xi = 1,\slots_l] \leq \sum\limits_{\slot' \in \slots_l}Pr(\overline{\Phi}_{l,\slot'}=1|\xi = 1,\slots_l) \cdot \beta x^{(l)}_{\slot'} \leq  \beta \frac{\knapsackcapacity}{l}$ and $\Ex[\cardinalitycapacity^{(l)}|\xi = 1,\slots_l] \leq \sum\limits_{\slot' \in \slots_l} Pr(\overline{\Phi}_{l,\slot'}=1|\xi = 1,\slots_l) \cdot y^{(l)}_{\slot'} \leq \frac{\cardinalitycapacity}{l}$. Since these two bounds apply to any $\slots_l\subseteq \slots$ and all subsets~$\slots_l$ are equally likely to be observed, we can remove the dependence on~$\slots_l$, i.e., we have $\Ex[\knapsackcapacity^{(l)}|\xi = 1] \leq  \beta \frac{\knapsackcapacity}{l}$ and $\Ex[\cardinalitycapacity^{(l)}|\xi = 1] \leq \frac{\cardinalitycapacity}{l}$.

In the case where no item is incorporated during the secretary phase (i.e., $\xi=1$), we  have~$\knapsackcapacity_{d\numitems+1} = \cardinalitycapacity_{d\numitems+1} = 0$,  i.e., $\online$ only starts to consume items at step~$dn+1$. Therefore, the expected utilization of the knapsack and cardinality capacity at the beginning of step~$l$ are upper-bounded by
\[
\Ex[\knapsackcapacity_{l}| \xi=1] 
    = 
        \sum_{k = dn + 1}^{l-1}\Ex[\knapsackcapacity^{(k)}| \xi=1] 
    \leq 
        \sum_{k = dn + 1}^{l-1}\beta\frac{\knapsackcapacity}{k} 
    \leq 
        \beta\int\limits_{dn}^{l-1}\frac{\knapsackcapacity}{k}dk 
    \leq 
        \beta\knapsackcapacity \ln{\frac{l-1}{dn}},
\]
and
\[
\Ex[\cardinalitycapacity_{l}| \xi=1] 
    = 
        \sum_{k = dn + 1}^{l-1}\Ex[\cardinalitycapacity^{(k)}| \xi=1] 
    \leq 
        \sum_{k = dn + 1}^{l-1}\frac{\cardinalitycapacity}{k} 
    \leq 
        \int\limits_{dn}^{l-1}\frac{\cardinalitycapacity}{k} dk
    \leq 
        \cardinalitycapacity \ln{\frac{l-1}{dn}},
\]
respectively, where the second-to-last inequalities on both expressions follow because~$\sum\limits_{i=a}^{b}\frac{1}{i} \leq \int\limits_{a-1}^{b}\frac{1}{i} di$ for $a, b\in \mathbb{Z}^+$. Finally, we use union bound and Markov's inequality to derive a lower bound for the probability with which~$\phi_{l} = 1$ given~$\xi=1$:
\begin{eqnarray*}
Pr\left[\phi_l = 1 | \xi=1 \right] 
    &=& 
        Pr\left[ (\knapsackcapacity_{l} < (1-\beta)\knapsackcapacity)  \cap (\cardinalitycapacity_{l} < \cardinalitycapacity) | \xi=1 \right] 
    \\
    &=& 
        1 - Pr\left[ (\knapsackcapacity_{l} \geq (1-\beta)\knapsackcapacity)  \cup 
            (\cardinalitycapacity_{l} \geq \cardinalitycapacity)  | \xi=1
            \right] 
    \\
    &\underbrace{\geq}_{\text{Union bound}}& 
        1 - Pr\left[ \knapsackcapacity_{l} \geq (1-\beta)\knapsackcapacity | \xi=1\right]  - Pr\left[ 
            \cardinalitycapacity_{l} \geq \cardinalitycapacity   | \xi=1
            \right] 
    \\
    &\underbrace{\geq}_{\text{Markov's inequality}}&
        1 - \frac{\Ex[\knapsackcapacity_{l}| \xi=1]}{(1-\beta)\knapsackcapacity} - 
        \frac{\Ex[\cardinalitycapacity_{l}| \xi=1]}{\cardinalitycapacity}
    \\
    &\geq& 
        1 - \left( \frac{\beta}{1 - \beta} + 1 \right)\ln{\frac{l-1}{dn}} 
    \\
    &=& 
        1 - \left( \frac{1}{1 - \beta} \right)\ln{\frac{l-1}{dn}}.  
\end{eqnarray*} 
This concludes the proof of Lemma~\ref{lemma:lb_l}. 
\medskip
\end{proof}

\begin{proof}{Proof of Lemma \ref{lemma:lb_knapsack}:} Based on Lemma \ref{lemma:simple} and Lemma~\ref{lemma:lb_l}, we have
\begin{eqnarray*}
        \sum_{l = d \numitems+1}^{\numitems} Pr[\overline{\Phi}_{l,\slot} = 1 | \xi = 1]
        Pr[\phi_{l} = 1 | \xi = 1] 
    &=& 
        \sum_{l = d \numitems+1}^{\numitems} \frac{1}{\numitems}
         Pr[\phi_{l} = 1 | \xi = 1] 
    \\
    &\geq&
         \frac{1}{\numitems}\sum\limits_{l = d \numitems + 1}^{\numitems} \left(1 - \left( \frac{1}{1 - \beta} \right)\ln{\frac{l-1}{dn}}\right) \\
    &=&
        (1 - d) -  \frac{1}{\numitems}\left( \frac{1}{1 - \beta} \right)\sum\limits_{l = d \numitems + 1}^{\numitems} \ln{\frac{l-1}{d\numitems}} \\
    &=&
        (1 - d) -  \frac{1}{\numitems}\left( \frac{1}{1 - \beta} \right)\sum\limits_{l = d \numitems}^{\numitems-1} \ln{\frac{l}{d\numitems}} 
\end{eqnarray*}

Since $\ln{\frac{l}{d\numitems}}$ is an increasing function in~$l$ and ~$\sum\limits_{l=a}^{b}\ln{\frac{l}{d\numitems}} \leq \int\limits_{a}^{b+1}\ln{\frac{l}{d\numitems}} dl$ for $a, b\in \mathbb{Z}^+$, we have
\[
\sum\limits_{l = d \numitems}^{\numitems-1} \ln{\frac{l}{d\numitems}} \leq 
\int\limits_{l = d \numitems}^{\numitems} \ln{\frac{l}{d\numitems}} dl
    = 
    \numitems\left(\ln{\frac{\numitems}{d\numitems}} - 1\right) - d\numitems\left(\ln{\frac{d\numitems}{d\numitems}} - 1\right) 
    = 
    d\numitems  - \numitems - \numitems\ln{d}.
\]
Therefore, as $\frac{1}{1-\beta} > 0$ for~$\beta \in (0,1)$, we have
\begin{eqnarray*}
\sum_{l = d \numitems+1}^{\numitems} Pr[\overline{\Phi}_{l,\slot} = 1 | \xi = 1]
        Pr[\phi_{l} = 1 | \xi = 1] 
    &\geq&
    (1 - d) + \left( \frac{1}{1 - \beta} \right)\left(1 -  d  + \ln{d} \right) 
    \\
    &=&
        (1 - d)\left(\frac{2 - \beta}{1 - \beta}  \right) + \left( \frac{1}{1 - \beta} \right) \ln{d}.     
\end{eqnarray*}
This concludes the proof of Lemma \ref{lemma:lb_knapsack}. 
\medskip
\end{proof}

\section{Proofs of Results in Section~\ref{sec:competitiveratio}}

\begin{proof}{Proof of Proposition \ref{prop:ratio_case1}:} From Lemma~\ref{lemma:sec}, $\online_{S}$ entirely packs the item~$\overline{\slot}$ with the largest totally utility with probability at least~$\overline{p}$. Moreover, item~$\overline{\slot}$ can be also packed by  $\online_{K}$. We note that any approximation algorithm~$\Approx_{\alpha}$ for the problem can capture the optimal solution for instances of type~$\Psi_1$ (i.e., fully packing the item~$\overline{\slot}$)  by incorporating an additional iteration to assess~$\utilities_{\slot}(\knapsackcapacity)$ for each~$\slot$ in~$\slots$ and identify~$\overline{\slot}$; this extra step can be executed in polynomial time and preserves the $\alpha$-approximation performance of the original algorithm. Therefore, if~$\online_{S}$ does not incorporate any item (i.e., $\xi = 1$), item $\overline{\slot}$ arrives at step $l \geq d \cdot \numitems + 1$ (i.e.,  $\overline{\Phi}_{l,\overline{\slot}}=1$), and sufficient capacity is left at step $l$ (i.e., $\phi_l=1$), $\online_{K}$ incorporates~$\beta \maxcapacity_{\overline{\slot}}=\beta \knapsackcapacity$ from $\overline{\slot}$ and obtain utility~$\utilities_{\overline{\slot}}(\beta\maxcapacity_{\overline{\slot}})=\utilities_{\overline{\slot}}(\beta\knapsackcapacity)$ at step $l$.  
Moreover, it follows from the concavity of~$\utilities_{\overline{\slot}}$ that $\utilities_{\overline{\slot}}(\beta\knapsackcapacity) \geq \beta \utilities_{\overline{\slot}}(\knapsackcapacity) = \beta\utilities^*$.

By combining the possible utility gains obtained from~$\overline{\slot}$ in the secretary and knapsack phases, we obtain
\begin{eqnarray}\label{eq:ratio_case1_proof}
    \Ex^{\Psi_1}_{\online}[\utilities(\vecx)] 
    &\geq& 
        \utilities^*\overline{p} +  \beta \utilities^* Pr[\xi = 1]\left(\sum_{l = d \numitems+1}^{\numitems} Pr[\overline{\Phi}_{l,\overline{\slot}} = 1 | \xi = 1]
        Pr[\phi_{l} = 1 | \xi = 1] \right)  \nonumber\\
    &\geq & 
        \utilities^*\left[c \ln{\frac{d}{c}} - o(1)\right] +  \beta \utilities^* \frac{c}{d}\left(\sum_{l = d \numitems+1}^{\numitems} Pr[\overline{\Phi}_{l,\overline{\slot}} = 1 | \xi = 1]
        Pr[\phi_{l} = 1 | \xi = 1] \right)  \nonumber\\
    &\geq & 
        \utilities^*\left[c \ln{\frac{d}{c}} - o(1)\right] +  \beta \utilities^* \frac{c}{d}\left((1 - d)\left(\frac{2 - \beta}{1 - \beta}  \right) + \left( \frac{1}{1 - \beta} \right) \ln{d} \right)  
    \nonumber\\
    &=& 
        \utilities^*\left[
     c \ln{\frac{d}{c}} + \frac{c}{d}\beta  \left((1 - d)\left(\frac{2 - \beta}{1 - \beta}  \right) + \left( \frac{1}{1 - \beta} \right) \ln{d}\right)- o(1)  
    \right]. \nonumber
\end{eqnarray}
The second inequality follows the results of Lemma \ref{lemma:sec} and Lemma \ref{lemma:cdbound}, and the third inequality follows from Lemma \ref{lemma:lb_knapsack}. This concludes the proof of Proposition \ref{prop:ratio_case1}. 
\medskip
\end{proof}

\begin{proof}{Proof of Proposition~\ref{prop:profit_knapsack}:}

We study the utility gain from both the secretary and knapsack phases. 

For the secretary phase, we focus on the expected utility obtained by~$\online_{S}$ from incorporating~$\overline{\slot}$, the most valuable item. We have shown that~$\online_{S}$ incorporates~$\overline{\slot}$ with probability $\overline{p}$ in Lemma \ref{lemma:sec} and obtains utility~$\utilities_{\overline{\slot}}(\maxcapacity_{\overline{\slot}})$. By definition, $R_{\slot}(x^*_{\slot}) \leq R_{\slot}(\maxcapacity_{\slot}) \leq
R_{\overline{\slot}}(\maxcapacity_{\overline{\slot}})$ for any~$\slot \in \slots$ and any optimal solution~$(\vecx^*,\vecy^*)$; as there are at most~$\cardinalitycapacity$ items being picked, we must have~$\utilities^*=\sum\limits_{\slot\in \slots:y^*_{\slot}=1}R_{\slot}(x^*_{\slot}) \leq\cardinalitycapacity \cdot
R_{\overline{\slot}}(\maxcapacity_{\overline{\slot}})$, and thus
$\utilities_{\overline{\slot}}(\maxcapacity_{\overline{\slot}}) \geq \frac{\utilities^*}{\cardinalitycapacity}$.

In the next part, we characterize the expected utility gain for the knapsack phase. This derivation is similar to the one presented in~\citet{kesselheim2018primal} while considering the impact of the approximation factor~$\alpha$ of~$\Approx_{\alpha}$.

    Let~$s^* = (\vecx^*,\vecy^*)$ be an optimal solution to~$\Problem{}(\slots)$. At any step~$l$, we consider a solution~$s^{(l)*} = (\vecx^{(l)*},\vecy^{(l)*})$ such that~$\left(x^{(l)*}_{\slot},y^{(l)*}_{\slot} \right) = \left( x^*_{\slot}, y^*_{\slot}\right)$ for~$\slot \in \slots_l$ and~$\left(x^{(l)*}_{\slot},y^{(l)*}_{\slot}\right) = (0,0)$ otherwise; let~$\utilities(\vecx^{(l)*})$    denote the utility of~$\vecx^{(l)*}$.    As the arrival sequences are uniformly distributed, each item has a probability of~$\frac{l}{n}$ to appear in $\slots_l$, and thus we have $\Ex[\utilities(\vecx^{(l)*})] = \frac{l}{n}\utilities^*$. 

    Let~$s^{(l)}= (\vecx^{(l)},\vecy^{(l)})$ be the solution identified by~$\Approx_{\alpha}(\slots_l)$. The utility of~$s^{(l)}$ is at least~$\alpha \utilities(\vecx^{(l)*})$ for any realization of~$\slots_l$ because~$\utilities(\vecx^{(l)*})$ is a lower bound for the optimal objective value of~$\Problem{}(\slots_l)$ and~$\Approx_{\alpha}(\slots_l)$ is an~$\alpha$-approximation algorithm to~$\Problem{}(\slots_l)$. Therefore, we must also have $\Ex[\utilities(\vecx^{(l)})] \geq \frac{l\alpha}{n}\utilities^*$.

    From the uniformity of the arrival orders, the expected utility of~$\slot_l$ in solution~$s^{(l)}$ is~$\frac{\Ex[\utilities(\vecx^{(l)})]}{l}$. Therefore, if the knapsack has enough residual capacity (i.e., $\phi_l=1$), the expected utility obtained by~$\online_{K}$ from item~$\slot_l$ is ~$\beta\frac{\Ex[\utilities(\vecx^{(l)})]}{l} \geq \frac{\beta\alpha}{n}\utilities^*$. This observation applies to all arrival position~$l$ in the knapsack phase. Therefore, the expected utility collected during the knapsack phase given~$\xi$ is
\begin{align*}
\Ex^{\Psi_2}_{\online_K}[\utilities(\vecx)|\xi = 1] 
    &= 
        \sum_{l = d \numitems+1}^{\numitems} \left( Pr[\phi_{l} = 1 | \xi = 1] \beta\frac{\Ex[\utilities(\vecx^{(l)})]}{l}\right) \\
    &\geq  
        \sum_{l = d \numitems+1}^{\numitems} \left( Pr[\phi_{l} = 1 | \xi = 1] \frac{\beta\alpha}{n}\utilities^*\right) \\
    &=
        (\beta\alpha\utilities^*)\sum_{l = d \numitems+1}^{\numitems}  \frac{1}{n}Pr[\phi_{l} = 1 | \xi = 1] \\
    &\geq 
        (\beta\alpha\utilities^*)\big[ (1 - d)\left(\frac{2 - \beta}{1 - \beta}  \right) + \left( \frac{1}{1 - \beta} \right)  \ln{d}\big].
\end{align*}
The last inequality follows from the fact that
$Pr[\overline{\Phi}_{l,\slot} = 1 | \xi = 1] = \frac{1}{\numitems}$ for every~$\slot$ in~$\slots$ (see
Lemma \ref{lemma:simple}) and Lemma \ref{lemma:lb_knapsack}, which give a lower bound for~$\sum\limits_{l = d \numitems+1}^{\numitems} Pr[\overline{\Phi}_{l,\slot} = 1 | \xi = 1]
        Pr[\phi_{l} = 1 | \xi = 1]$.  

By combining the expected returns from both phases, we obtain 
\begin{eqnarray*}
        \Ex^{\Psi_2}_{\online}[\utilities(\vecx)] 
    &\geq& 
        \frac{\utilities^*}{\cardinalitycapacity}\overline{p} + 
        Pr[\xi = 1]\left[ (1 - d)\left(\frac{2 - \beta}{1 - \beta}  \right) + \left( \frac{1}{1 - \beta} \right)  \ln{d}\right]\alpha\beta\utilities^*\\
    &\geq& 
        \frac{\alpha \utilities^*}{\cardinalitycapacity}\left(c \ln{\frac{d}{c}} - o(1)\right) + 
        \frac{c}{d}  \left((1 - d)\left(\frac{2 - \beta}{1 - \beta}  \right) + \left( \frac{1}{1 - \beta} \right) \ln{d}  \right)\alpha\beta\utilities^*
     \\
    &=& 
        \alpha \utilities^*\left[\frac{c}{\cardinalitycapacity} \ln{\frac{d}{c}} + \frac{c}{d}\beta  \left((1 - d)\left(\frac{2 - \beta}{1 - \beta}  \right) + \left( \frac{1}{1 - \beta} \right) \ln{d}  \right)- \frac{1}{\cardinalitycapacity}o(1)
    \right].\\
        &\geq& 
        \alpha \utilities^*\left[\frac{c}{\cardinalitycapacity} \ln{\frac{d}{c}} + \frac{c}{d}\beta  \left((1 - d)\left(\frac{2 - \beta}{1 - \beta}  \right) + \left( \frac{1}{1 - \beta} \right) \ln{d}  \right)- o(1)
    \right],    
\end{eqnarray*}
where the second inequality follows because~$\alpha \leq 1$. This concludes the proof of Proposition~\ref{prop:profit_knapsack}.  
 \end{proof}

\begin{proof}{Proof of Corollary\ref{cor:ratio2}:} 

The derivations are similar to the ones presented for the general case but require a few adaptations. We define $\overline{\phi}_l=1$ to indicate that the condition $\knapsackcapacity_{l} < (1 - \beta) \knapsackcapacity$ is satisfied upon the arrival of the~$l$-th item. Lemma~\ref{lemma:lb_l2} is the equivalent of Lemma~\ref{lemma:lb_l} when~$\cardinalitycapacity > (1-d)\numitems$ and its proof follows the proof for Lemma~\ref{lemma:lb_l}. Similarly, Lemma~\ref{lemma:lb_knapsack2} is the adaption of Lemma~\ref{lemma:lb_knapsack}.

\begin{lemma}\label{lemma:lb_l2}
    For every~$l \geq dn+1$, if $\cardinalitycapacity \geq d \cdot \numitems$,
    we have
    \begin{equation*}\label{eq:lb_l2}
         Pr\left[ \overline{\phi}_l  = 1 | \xi =1  \right]
     \geq  1 - \left( \frac{\beta}{1 - \beta} \right)\ln{\frac{l-1}{dn}}.   
        \end{equation*}
    
\end{lemma}
\begin{proof}{Proof of Lemma \ref{lemma:lb_l2}:} In this case, we only need the upper bound for~$\Ex[\knapsackcapacity_{l}| \xi=1]$, given by $\Ex[\knapsackcapacity_{l}| \xi=1] 
    \leq 
        \beta\knapsackcapacity \ln{\frac{l-1}{dn}}$, and use Markov's inequality to derive a lower bound for~$Pr\left[ \overline{\phi}_l  = 1 | \xi =1  \right]$:
\begin{eqnarray*}
Pr\left[\overline{\phi}_l = 1 | \xi=1 \right] 
    &=& 
        Pr\left[ (\knapsackcapacity_{l} < (1-\beta)\knapsackcapacity)   | \xi=1 \right] 
    \\
    &=& 
        1 - Pr\left[ (\knapsackcapacity_{l} \geq (1-\beta)\knapsackcapacity)   | \xi=1
            \right] 
    \\
    &\underbrace{\geq}_{\text{Markov's inequality}}&
        1 - \frac{\Ex[\knapsackcapacity_{l}| \xi=1]}{(1-\beta)\knapsackcapacity} 
    \\
    &\geq& 
        1 - \left( \frac{\beta}{1 - \beta} \right)\ln{\frac{l-1}{dn}}.  
\end{eqnarray*} 
This concludes the proof of Lemma~\ref{lemma:lb_l2}. 
\medskip
\end{proof}

\begin{lemma}\label{lemma:lb_knapsack2} For every~$\slot$ in~$\slots$, if~$\cardinalitycapacity \geq d \numitems$, we have
  \begin{equation*}
  \sum_{l = d \numitems+1}^{\numitems} Pr[\overline{\Phi}_{l,\slot} = 1 | \xi = 1]
        Pr[\overline{\phi}_{l} = 1 | \xi = 1]\geq
(1 - d)\left(\frac{1}{1 - \beta}  \right) + \left( \frac{\beta}{1 - \beta} \right) \ln{d}.
        \end{equation*}
\end{lemma}
\begin{proof}{Proof of Lemma \ref{lemma:lb_knapsack2}:} Based on Lemma \ref{lemma:simple} and Lemma~\ref{lemma:lb_l2}, we have
\begin{eqnarray*}
        \sum_{l = d \numitems+1}^{\numitems} Pr[\overline{\Phi}_{l,\slot} = 1 | \xi = 1]
        Pr[\overline{\phi}_{l} = 1 | \xi = 1] 
    &=& 
        \sum_{l = d \numitems+1}^{\numitems} \frac{1}{\numitems}
         Pr[\overline{\phi}_{l} = 1 | \xi = 1]  
    \\
    &\geq&
         \frac{1}{\numitems}\sum\limits_{l = d \numitems + 1}^{\numitems} \left(1 - \left( \frac{\beta}{1 - \beta} \right)\ln{\frac{l-1}{dn}}\right) \\
    &=&
        (1 - d) -  \frac{1}{\numitems}\left( \frac{\beta}{1 - \beta} \right)\sum\limits_{l = d \numitems + 1}^{\numitems} \ln{\frac{l-1}{d\numitems}} \\
    &=&
        (1 - d) -  \frac{1}{\numitems}\left( \frac{\beta}{1 - \beta} \right)\sum\limits_{l = d \numitems}^{\numitems-1} \ln{\frac{l}{d\numitems}} \\
    &\geq&
        (1 - d) + \left( \frac{\beta}{1 - \beta} \right)\left(1 -  d  + \ln{d} \right) 
        \\
    &=&
        (1 - d)\left(\frac{1}{1 - \beta}  \right) + \left( \frac{\beta}{1 - \beta} \right) \ln{d}.     
\end{eqnarray*}
This concludes the proof of Lemma~\ref{lemma:lb_knapsack2}.  
\medskip
\end{proof}

Proposition~\ref{prop:ratio_case1_2} and Proposition~\ref{prop:profit_knapsack_2} are the adaptations of Proposition~\ref{prop:ratio_case1} and  Proposition~\ref{prop:profit_knapsack}, respectively.

\begin{proposition}
    \label{prop:ratio_case1_2}
The expected utility gain $\Ex^{\Psi_1}_{\online}[\utilities(\vecx)]$ for instances of type~$\Psi_1$ if~$\cardinalitycapacity \geq d \numitems$
is
\begin{equation*}
    \Ex^{\Psi_1}_{\online}[\utilities(\vecx)] \geq
        \utilities^*\left[
     c \ln{\frac{d}{c}} + \frac{c}{d}\beta  \left((1 - d)\left(\frac{1}{1 - \beta}  \right) + \left( \frac{\beta}{1 - \beta} \right) \ln{d}  \right)- o(1)  
    \right].
\end{equation*}
 \end{proposition}
\begin{proof}{Proof of Proposition \ref{prop:ratio_case1_2}:} By replacing~$\phi_l$ with~$\overline{\phi}_l$ in the proof of Proposition~\ref{prop:ratio_case1}, we obtain the following lower bound:
\begin{eqnarray}\label{eq:ratio_case1_proof_2}
    \Ex^{\Psi_1}_{\online}[\utilities(\vecx)] 
    &\geq& 
        \utilities^*\overline{p} +  \beta \utilities^* Pr[\xi = 1]\left(\sum_{l = d \numitems+1}^{\numitems} Pr[\overline{\Phi}_{l,\overline{\slot}} = 1 | \xi = 1]
        Pr[\overline{\phi}_{l} = 1 | \xi = 1] \right)  \nonumber\\
    &\geq & 
        \utilities^*\left[c \ln{\frac{d}{c}} - o(1)\right] +  \beta \utilities^* \frac{c}{d}\left(\sum_{l = d \numitems+1}^{\numitems} Pr[\overline{\Phi}_{l,\overline{\slot}} = 1 | \xi = 1]
        Pr[\overline{\phi}_{l} = 1 | \xi = 1] \right)  \nonumber\\
    &\geq & 
        \utilities^*\left[c \ln{\frac{d}{c}} - o(1)\right] +  \beta \utilities^* \frac{c}{d}\left((1 - d)\left(\frac{1}{1 - \beta}  \right) + \left( \frac{\beta}{1 - \beta} \right) \ln{d} \right)  
    \nonumber\\
    &=& 
        \utilities^*\left[
     c \ln{\frac{d}{c}} + \frac{c}{d}\beta  \left((1 - d)\left(\frac{1}{1 - \beta}  \right) + \left( \frac{\beta}{1 - \beta} \right) \ln{d}\right)- o(1)  
    \right]. \nonumber 
\end{eqnarray}
This concludes the proof of Proposition~\ref{prop:ratio_case1_2}.  
\medskip
\end{proof}

\begin{proposition}
\label{prop:profit_knapsack_2}
The expected utility gain $\Ex^{\Psi_2}_{\online}[\utilities(\vecx)]$ for instances of type~$\Psi_2$ if~$\cardinalitycapacity \geq d \numitems$ is
\begin{equation*}
        \Ex^{\Psi_2}_{\online}[\utilities(\vecx)] 
    \geq 
        \alpha \utilities^*\left[\frac{c}{\cardinalitycapacity} \ln{\frac{d}{c}} + \frac{c}{d}\beta  \left((1 - d)\left(\frac{1}{1 - \beta}  \right) + \left( \frac{\beta}{1 - \beta} \right) \ln{d}  \right)- o(1)
    \right].
\end{equation*}
\end{proposition}
\begin{proof}{Proof of Proposition~\ref{prop:profit_knapsack_2}:} By replacing~$\phi_l$ with~$\overline{\phi}_l$ in the proof of Proposition~\ref{prop:profit_knapsack}, we obtain 
\begin{align*}
\Ex^{\Psi_2}_{\online_K}[\utilities(\vecx)|\xi = 1] 
    &= 
        \sum_{l = d \numitems+1}^{\numitems} \left( Pr[\overline{\phi}_{l} = 1 | \xi = 1] \beta\frac{\Ex[\utilities(\vecx^{(l)})]}{l}\right) \\
    &\geq  
        \sum_{l = d \numitems+1}^{\numitems} \left( Pr[\overline{\phi}_{l} = 1 | \xi = 1] \frac{\beta\alpha}{n}\utilities^*\right) \\
    &=
        (\beta\alpha\utilities^*)\sum_{l = d \numitems+1}^{\numitems}  \frac{1}{n}Pr[\overline{\phi}_{l} = 1 | \xi = 1] \\
    &\geq 
        (\beta\alpha\utilities^*)\big[ (1 - d)\left(\frac{1}{1 - \beta}  \right) + \left( \frac{\beta}{1 - \beta} \right)  \ln{d}\big].
\end{align*}
By combining the expected returns from both phases, we obtain 
\begin{eqnarray*}
        \Ex^{\Psi_2}_{\online}[\utilities(\vecx)] 
    &\geq& 
        \frac{\utilities^*}{\cardinalitycapacity}\overline{p} + 
        Pr[\xi = 1]\left[ (1 - d)\left(\frac{1}{1 - \beta}  \right) + \left( \frac{1}{1 - \beta} \right)  \ln{d}\right]\alpha\beta\utilities^*\\
    &\geq& 
        \frac{\alpha \utilities^*}{\cardinalitycapacity}\left(c \ln{\frac{d}{c}} - o(1)\right) + 
        \frac{c}{d}  \left((1 - d)\left(\frac{1}{1 - \beta}  \right) + \left( \frac{1}{1 - \beta} \right) \ln{d}  \right)\alpha\beta\utilities^*
     \\
    &=& 
        \alpha \utilities^*\left[\frac{c}{\cardinalitycapacity} \ln{\frac{d}{c}} + \frac{c}{d}\beta  \left((1 - d)\left(\frac{1}{1 - \beta}  \right) + \left( \frac{\beta}{1 - \beta} \right) \ln{d}  \right)- \frac{1}{\cardinalitycapacity}o(1)
    \right].\\
        &\geq& 
        \alpha \utilities^*\left[\frac{c}{\cardinalitycapacity} \ln{\frac{d}{c}} + \frac{c}{d}\beta  \left((1 - d)\left(\frac{1}{1 - \beta}  \right) + \left( \frac{\beta}{1 - \beta} \right) \ln{d}  \right)- o(1)
    \right].     
\end{eqnarray*}
This concludes the proof of Proposition~\ref{prop:profit_knapsack_2}.  
\medskip
 \end{proof}

Finally, we derive the competitive ratio by deriving local optimal solutions to the following problem with $\cardinalitycapacity \rightarrow \infty$: 
\[
\max\limits_{0 \leq c \leq d \leq 1, 0 < \beta \leq 1, d \geq \frac{\cardinalitycapacity}{n}} \quad 
\frac{c}{\cardinalitycapacity} \ln{\frac{d}{c}} + \frac{c}{d}\beta  \left((1 - d)\left(\frac{1}{1 - \beta}  \right) + \left( \frac{\beta}{1 - \beta} \right) \ln{d}  \right).
\]
By setting~$c = d = \beta \approx 0.431$, we have that~$\online$ is $6.401\alpha$-competitive if~$\cardinalitycapacity \geq 0.569 \numitems$. This concludes the proof of Corollary~\ref{cor:ratio2}. 
\end{proof}

\end{APPENDICES}

\end{document}